\documentclass{article}

\usepackage{amsmath,amssymb,amsfonts,amsthm}
\usepackage[capitalize]{cleveref}
\usepackage{algorithm,algpseudocode}

\DeclareMathOperator*{\argmin}{arg\,min}
\usepackage{graphicx}
\usepackage{textcomp}
\usepackage{xcolor}
\usepackage{mathtools}
\usepackage{thmtools,thm-restate}
\usepackage{subfigure}

\usepackage{flushend}
\allowdisplaybreaks
\usepackage{wrapfig}
\usepackage{array}
\usepackage[margin=1in]{geometry}
\usepackage{breakcites}
\newcolumntype{P}[1]{>{\centering\arraybackslash}p{#1}}

\usepackage{todonotes}

\newcommand{\andd}{\bigwedge}
\newcommand{\greed}{\mathcal{G}_{S}}

\newcommand{\majorby}{\preceq}

\newcommand{\opt}{\textsc{OPT}_{S}}

\newcommand{\mat}[1]{\mathbf{#1}}

\newcommand{\coupleProb}{\textsc{MEC}}%

\newcommand{\entropyProb}{\textsc{Compute-Entropy}}
\newcommand{\profile}{\textsc{Profile}}
\newcommand{\iprof}{\textsc{Inv-Prof}}
\newcommand{\sketch}{\textsc{Sketch}}
\newcommand{\isketch}{\textsc{Inv-Sketch}}
\newcommand{\majorpro}{\textsc{Major-Profile}}

\newcommand{\rmass}{\textsc{Rem-Mass}}
\newcommand{\mass}{\textsc{Mass}}
\newcommand{\rmasssim}{\textsc{Rem-Mass}^{\textsc{simple}}}
\newcommand{\rmassadv}{\textsc{Rem-Mass}^{\textsc{advanced}}}
\newcommand{\Fc}{F_{\textrm{cost}}}
\newcommand{\fc}{f_{\textrm{cost}}}
\newcommand{\fu}{f_{\textrm{unit}}}
\renewcommand{\lg}{\log}

\newtheorem{theorem}{Theorem}[section]
\newtheorem{lemma}[theorem]{Lemma}

\newtheorem{definition}[theorem]{Definition}
\newtheorem{corollary}[theorem]{Corollary}

\newtheorem{remark}{Remark}

\newcommand{\BP}{\mathbf P}
\newcommand{\p}{\mathbf p}
\newcommand{\G}{\mathcal G}

\newcommand{\vphi}{\varphi}
\newcommand{\paren}[1]{\left( #1 \right)}

\begin{document}

\title{
Minimum-Entropy Coupling Approximation Guarantees \\Beyond the Majorization Barrier
}

\date{}

\author{
  Spencer Compton \\
  Stanford University \\
  MIT-IBM Watson AI Lab \\
  comptons@stanford.edu
  \and
  Dmitriy Katz\\
  MIT-IBM Watson AI Lab \\
  IBM Research\\
  dkatzrog@us.ibm.com
  \and
  Benjamin Qi\\
  MIT\\
  bqi343@gmail.com
  \and
  Kristjan Greenewald\\
  MIT-IBM Watson AI Lab \\
  IBM Research\\
  kristjan.h.greenewald@ibm.com
  \and
  Murat Kocaoglu\\
  Purdue University\\
  mkocaoglu@purdue.edu
}

\maketitle

\begin{abstract}
   Given a set of discrete probability distributions, the minimum entropy coupling is the minimum entropy joint distribution that has the input distributions as its marginals. This has immediate relevance to tasks such as entropic causal inference for causal graph discovery and bounding mutual information between variables that we observe separately. Since finding the minimum entropy coupling is NP-Hard, various works have studied approximation algorithms. The work of \cite{compton2022tighter} shows that the greedy coupling algorithm of \cite{kocaoglu2017entropic} is always within $\log_2(e)\approx1.44$ bits of the optimal coupling. 
   Moreover, they show that it is impossible to obtain a better approximation guarantee using the majorization lower bound that all prior works have used: thus establishing a \emph{majorization barrier}.
   In this work, we break the majorization barrier by designing a stronger lower bound that we call the \emph{profile method}. Using this profile method, we are able to show that the greedy algorithm is always within $\log_2(e)/e \approx 0.53$ bits of optimal for coupling two distributions (previous best-known bound is within 1 bit), and within $\frac{1 + \log_2(e)}{2}\approx 1.22$ bits for coupling any number of distributions (previous best-known bound is within 1.44 bits). We also examine a generalization of the minimum entropy coupling problem: \emph{Concave Minimum-Cost Couplings}. We are able to obtain similar guarantees for this generalization in terms of the concave cost function. 
   Additionally, we make progress on the open problem of \cite{kovavcevic2015entropy} regarding NP membership of the minimum entropy coupling problem by
   showing that any hardness of minimum entropy coupling beyond NP comes from the difficulty of computing arithmetic in the complexity class NP.
   Finally, we present exponential-time algorithms for computing the \emph{exactly optimal} solution. We experimentally observe that our new profile method lower bound is not only helpful for analyzing the greedy approximation algorithm, but also for improving the speed of our new backtracking-based exact algorithm. 
\end{abstract}

\section{Introduction}
Entropy, particularly in the context of its maximization, has found wide application in statistics and machine learning, e.g., in the maximum entropy principle for estimation \cite{levine1979maximum,nigam1999using}, entropic optimal transport \cite{cuturi2013sinkhorn}, and, more generally, entropic regularization \cite{grandvalet2006entropy,niu2014information,haarnoja2018soft}. 

While much more challenging to optimize due to the concavity of entropy, there recently has been increasing interest in the \emph{minimization} of entropy, particularly in the context of the {minimum entropy coupling} (MEC) problem.
Recent applications of the MEC include the entropic causal inference framework \cite{kocaoglu2017entropic,javidian2021quantum,compton2022entropic}; communications \cite{sokota2022communicating,anonymous2023perfectly}; random number generation (discussed in \cite{li2021efficient}); functional representation (discussed in \cite{cicalese2019minimum}); and dimensionality reduction \cite{vidyasagar2012metric, cicalese2016approximating}. 

The minimum entropy coupling problem seeks to answer the question: Given a set $S$ of $m$ marginal (discrete-valued) distributions, each with at most $n$ states, what is the minimum-entropy joint distribution (i.e. coupling) that explains the marginals? While simple to state and appealing from an information-theoretic viewpoint, algorithmic optimization over a concave cost (entropy) is often difficult (e.g. \cite{cardinal2008tight}). In particular, computing the minimum entropy coupling is NP-hard \cite{kovavcevic2015entropy}. Despite this challenge, the minimum entropy coupling problem continues to attract interest and has seen an increasing number of applications.

Note that the entropy of the coupling is related to the mutual information between the marginal variables, through the identity $I(X; Y) = H(X) + H(Y) - H(X,Y)$. In particular, the minimum entropy coupling is the coupling that \emph{maximizes} the mutual information, i.e. it can be viewed as the \emph{maximum mutual information} coupling. The perspective of maximizing mutual information has been very influential in machine learning \cite{viola1997alignment,torkkola2003feature,tschannen2019mutual,sun2020infograph}. In discrete-valued settings, this coupling-based upper bound on mutual information can be of interest, as it provides a tighter bound than the traditional bound by the entropy of the marginal distributions $\min(H(X),H(Y))$. Accordingly, the minimum entropy coupling enables one to upper bound the mutual information between variables that are observed separately. Additionally, couplings have immediate relevance to optimal transport (OT), where the set of all couplings is called the transport polytope. As remarked by \cite{cicalese2019minimum}, the minimum entropy coupling is the minimum entropy element in the transport polytope.

The entropic causal inference framework proposed in  \cite{kocaoglu2017entropic}, and further developed in \cite{kocaoglu2017isit,compton2020entropic,javidian2021quantum,compton2022entropic} addresses the problem of orienting causal graphs when only observational data is available. In such a regime, the data is not sufficient to orient the graph and additional assumptions are required. Existing assumption frameworks (e.g. additive noise models) were not amenable to settings such as those with categorical variables, hence \cite{kocaoglu2017entropic} introduced an empirically and intuitively motivated (e.g. Occam's razor) hypothesis that the true causal direction between a pair of categorical variables often has the smallest exogenous noise entropy associated with its generative model. Finding this exogenous noise entropy from the observed joint distribution was shown to correspond to finding the minimum entropy coupling of a (possibly large) set of distributions. The later work of \cite{compton2020entropic} showed an identifiability result that if the true causal direction has low entropy exogeneous noise, then the reverse (anticausal) direction will have large entropy exogenous noise with high probability. An extension to general multi-variable causal directed acyclic graphs was presented in \cite{compton2022entropic}. In all of these works, solving or approximating the minimum entropy coupling problem is key to applying the entropic causal framework, and in practice the greedy algorithm of \cite{kocaoglu2017entropic} is used to approximate the minimum entropy coupling. Improved guarantees on the accuracy of the greedy algorithm for many distributions such as our \cref{thm:1.22}, or improved algorithmic approaches, are therefore useful for improving the applicability and trustworthiness of entropic causality.

The recent work of \cite{sokota2022communicating} studies Markov coding games: a generalization of source coding and referential games. At a high level, the crux of their approach is to combine reinforcement learning algorithms with minimum entropy coupling algorithms, in an effort to create an agent that can communicate well through Markov decision process trajectories. As an example, they design an agent that can efficiently communicate images just through its actions in the video game Pong and still play the game well. This agent extensively uses the coupling algorithm of \cite{cicalese2019minimum} to couple pairs of distributions. The work of \cite{anonymous2023perfectly} also uses minimum entropy coupling for communication, in the context of steganography: securely encoding secret information concealed in seemingly-regular text. They show that under an information-theoretic model of steganography, a procedure is secure if and only if it corresponds to a coupling. Moreover, the minimum entropy coupling corresponds to the maximally efficient secure procedure. They demonstrate strong performance for the minimum entropy coupling-based approach for perfectly secure steganography in experiments using GPT-2 \cite{radford2019language} and WaveRNN \cite{kalchbrenner2018efficient} as communication channels. In doing so, they extensively utilize the greedy coupling algorithm of \cite{kocaoglu2017entropic} to couple pairs of distributions. Our result of \cref{thm:0.53} directly improves the approximation guarantees for coupling pairs of distributions: a crucial algorithm for both works.

The work of \cite{vidyasagar2012metric} directly computes the minimum entropy coupling (by a different name) of two distributions in their metric for dimension-reduction of stochastic processes. Accordingly, our result of \cref{thm:0.53} directly improves the best-known approximation guarantee for efficiently computing this metric. The work of \cite{cicalese2016approximating} likewise uses minimum entropy coupling for dimension reduction of probability distributions.

A variety of additional applications of minimum entropy couplings are discussed in \cite{cicalese2019minimum,li2021efficient}.

\paragraph{Our contributions.} Many recent works have focused on proving approximation guarantees for the minimum entropy coupling problem. A unifying property of \cite{cicalese2019minimum,rossi2019greedy,li2021efficient,compton2022tighter} is that they all show their approximation guarantee in relation to the same lower bound: \emph{the majorization lower bound}. Moreover, \cite{compton2022tighter} shows that better guarantees for general $m$ in relation to the majorization lower bound are impossible: thus establishing the majorization barrier.

In this work, we move away from the majorization strategy and introduce a new method for lower bounding the coupling entropy, which we call the ``profile method.'' This approach allows us to obtain stronger bounds in both the case of coupling two distributions and for coupling an arbitrary number of distributions, as is shown in \cref{table:guarantees}.

We note that for all applications, the task of choosing what coupling algorithm to use was previously unclear. Prior to our work, there were three distinct algorithms that all were shown to have a 1 bit additive approximation guarantee for coupling two distributions \cite{cicalese2019minimum,rossi2019greedy,li2021efficient}. In this sense, there was no theoretical rationale for using one algorithm over another. For example, the recent work of \cite{sokota2022communicating} used the coupling algorithm of \cite{cicalese2019minimum}, while the work of \cite{anonymous2023perfectly} used the coupling algorithm of \cite{kocaoglu2017entropic}. Our work provides clarity by improving the approximation guarantee for the greedy coupling of \cite{kocaoglu2017entropic} and providing counter-examples that show the algorithms of \cite{cicalese2019minimum,li2021efficient} do not match its guarantees (see \cref{section:gaps}).

We also observe that our profile method is not limited to minimum entropy coupling, as we extend our approach to general concave minimum-cost couplings. 

In addition to new (highly efficient to compute) lower bounds on the coupling, we also examine the problem of exact computation of the minimum entropy coupling problem. First, we make progress on the open problem of \cite{kovavcevic2015entropy} regarding NP membership of the minimum entropy coupling problem by providing an NP-reduction from minimum entropy coupling to the problem of simply calculating the Shannon entropy of a distribution. In essence, this shows that any hardness of minimum entropy coupling beyond NP comes from the difficulty of computing arithmetic in the complexity class NP. 

We then prove a structural result that enables faster (but still exponential-time) algorithms for exactly coupling two distributions. Namely, we propose a dynamic programming and backtracking algorithm that are significantly faster than the naive baseline. We observe that our profile method lower bound mentioned above is not only helpful for analysis, but also for speeding up the backtracking algorithm.

We summarize the contributions of our paper as follows:
\begin{enumerate}
    \item A new profile method for lower bounding the entropy of the optimal coupling.
    \item Approximation guarantee for coupling two distributions: $H(\greed) \le H(\opt) + \frac{\log_2(e)}{e} \approx H(\opt) + 0.53$, where $\greed$ is the greedy coupling.
    \item Approximation guarantee for coupling an arbitrary number of distributions: $H(\greed) \le H(\opt) + \frac{1 + \log_2(e)}{2} \approx H(\opt) + 1.22$.
    \item An NP-reduction from minimum entropy coupling to the problem of simply calculating the Shannon entropy of a distribution. 
    \item %
    Experimentally faster (still exponential-time) algorithms for exactly coupling two distributions.
    \item Experimental results showing that the profile method lower bound is not only helpful for analysis, but also for speeding up the backtracking algorithm.
\end{enumerate}
\begin{table}[htbp] 
\centering
\caption{Best-Known Additive Approximation Guarantees} \label{table:guarantees}
\scriptsize
\begin{tabular}{|P{0.5in}|P{1.6in}|P{1.4in}|}
\hline
\textbf{} & \textbf{Prior Work}& \textbf{This Work} \\
\hline
{$m=2$}& $1$ \cite{cicalese2019minimum,rossi2019greedy,li2021efficient} & \vspace{-0.05cm} ${\frac{1}{e}\log_2(e) \approx 0.53}$ \\
\hline
{General $m$}& \vspace{-0.05cm} $\log_2(e) \approx 1.44$ \cite{compton2022tighter} & \vspace{-0.05cm} ${\frac{1}{2}(\log_2(e) + 1)\approx 1.22}$ \\
\hline
\multicolumn{3}{l}{\cref{section:small-m} contains a more complete table with additional results for small $m$.}
\end{tabular}
\label{tab1}
\end{table}

\section{Related Work and Background}\label{section:related-background}

\textit{Notation:} $\ln$ and $\lg$ denote $\log_e$ and $\log_2$, respectively. $H$ denotes Shannon entropy.  Throughout the paper, assume the states of any probability distribution $p$ are sorted in non-increasing order of size; that is, $p(1)\ge p(2) \ge \dots \ge p(|p|)$. We say that a probability distribution is \textit{possibly partial} if the sum of its sizes may be less than one. $S$ denotes a set of (possibly partial) probability distributions, all with the same total mass $\mass(S)$. $\opt$ denotes an optimal coupling of $S$. 

\textit{Couplings:} A coupling $\mathcal{C}$ is a joint distribution over $m$ input distributions $p_1, \dots, p_m$ such that the marginals of $\mathcal{C}$ are the input distributions. In particular, each state of a coupling $\mathcal{C}$ maps to exactly one state for each of the input distributions, where $\mathcal{C}(i_1, i_2, \dots, i_{m})$ maps to $p_1(i_1), p_2(i_2), \dots, p_m(i_m)$. Thus, the marginal condition means that the total mass of a coupling's states mapping to a particular state $i_k$ of an input distribution $p_k$ must be equal to the mass of $p_k(i_k)$, or equivalently $\mathcal{C}(\cdot, \dots, i_k, \cdot) = p_k(i_k)$ for all $k,i_k$. For simplicity of presentation, we will denote every input distribution $p_i$ as having $n$ states $p_i(1),\dots,p_i(n)$. However, we emphasize that all our results hold for coupling distributions with different numbers of states; we can simply ``pad'' the smaller distributions by appending states with probability 0.

\textit{Greedy Coupling:} Our work will analyze the greedy coupling algorithm of \cite{kocaoglu2017entropic}, formally described in \cref{alg:greedy}. In words, the algorithm repeatedly makes a greedy choice for the next state of the coupling, where it creates a state that maps to the maximally large state in each distribution, and has size equal to the minimum of these maximums.
\cite{kocaoglu2017isit} showed this algorithm is locally optimal, \cite{rossi2019greedy} showed it finds a coupling within 1 bit of the optimal entropy for coupling two distributions, and \cite{compton2022tighter} showed it is within $\log(e) \approx 1.44$ bits for coupling arbitrary numbers of distributions. We use $\greed$ to denote the sequence of sizes of states the greedy algorithm produces. It is known that $|\greed| \le nm - (m-1)$ and $\greed$ is non-increasing.

\begin{algorithm}[ht!]
\begin{small}
    \caption{Greedy Coupling \cite{kocaoglu2017entropic}}
   \label{alg:greedy}
\begin{algorithmic}[1]
    \State {\bfseries Input:} Marginal distributions of $m$ variables each with $n$ states, in matrix form $\mat{M} = [p_1^T;p_2^T;...;p_m^T]$.
    \State $\greed = [ \hspace{0.1in}]$
    \State Sort each row of $\mat{M}$ in non-increasing order.
    \State Find minimum of maximum of each row: $r\leftarrow \min_i(p_i(1))$
     \While  {$r>0$} 
     \State Append $r$ as the next state of $\greed$.
     \State Update maximum of each row: $p_i(1)\leftarrow p_i(1)-r, \forall i$
     \State Sort each row of $\mat{M}$ in non-increasing order.
     \State $r\leftarrow \min_i(p_i(1))$
    \EndWhile
    \State \Return $\greed$.
\end{algorithmic}
\end{small}
\end{algorithm}

\textit{Majorization:}
A distribution $p$ is majorized by a distribution $q$ (denoted by $p \majorby q$) if and only if $\sum_{j=1}^i p(j) \le \sum_{j=1}^i q(j)$ for all $i \in \{1,\dots,|q|\}$ \cite{marshall1979inequalities}. It is known that any valid coupling must be majorized by all distributions in $S$, meaning $\opt \majorby p,  \forall p \in S$ \cite{cicalese2019minimum}. Consider the partial ordering induced by majorization. Let $\andd S$ denote the ``maximal'' distribution in this partial ordering such that $\andd S \majorby p, \, \forall p \in S$, specifically, $\andd S(i) = \left({\min_{p \in S} \sum_{j=1}^{i} p(j)}\right) - \andd S(i-1)$. It can be shown that if $p \majorby q$ then $H(q) \le H(p)$ (this holds for any concave function), hence by definition $H(\opt) \ge H(\andd S)$. We thus have $H(\andd S)$ as a lower bound to $H(\opt)$, which is itself upper bounded by the greedy algorithm. This fact can be used to provide an approximation guarantee for the greedy algorithm, by bounding the gap between $H(\andd S)$ and the entropy returned by the greedy algorithm. \cite{compton2022tighter} shows that this gap is at most $\log_2(e) \approx 1.44$ bits, and shows that this $\log_2(e)$ gap can be achieved for general $m$. Hence the bound is tight in the sense that no better approximation guarantee is possible for greedy coupling in terms of this $H(\andd S)$ gap for general $m$.

\section{Novel Lower Bound: the Profile Method}\label{section:novel-lower-bound}

As mentioned above, prior works show guarantees for minimum entropy coupling in relation to the majorization lower bound $H(\andd S)$. Our work will introduce a stronger lower bound that will enable us to show stronger approximation guarantees that break the majorization barrier.

\begin{figure*}[ht!]
\begin{center}
\subfigure[{Constraints of $p_1$ ([0.5,0.4,0.1])}]{\label{fig:profile-A}\includegraphics[width=0.45\textwidth]{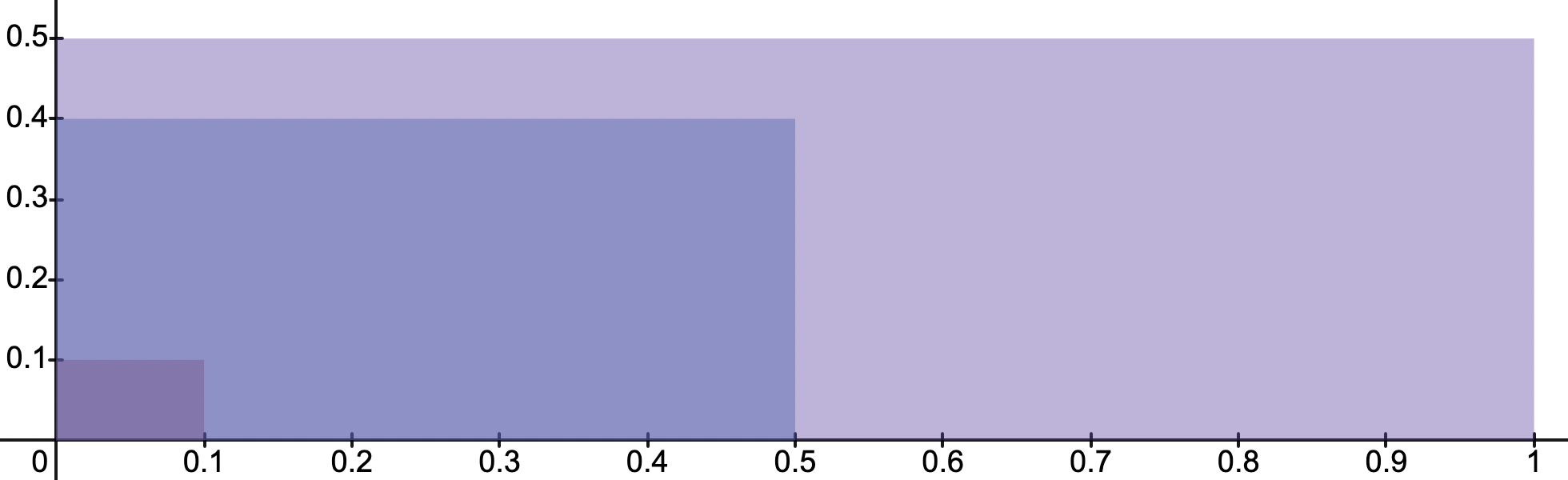}}\hspace{0.3in}
\subfigure[Binding constraints of $p_1$]{\label{fig:profile-B}\includegraphics[width=0.45\textwidth]{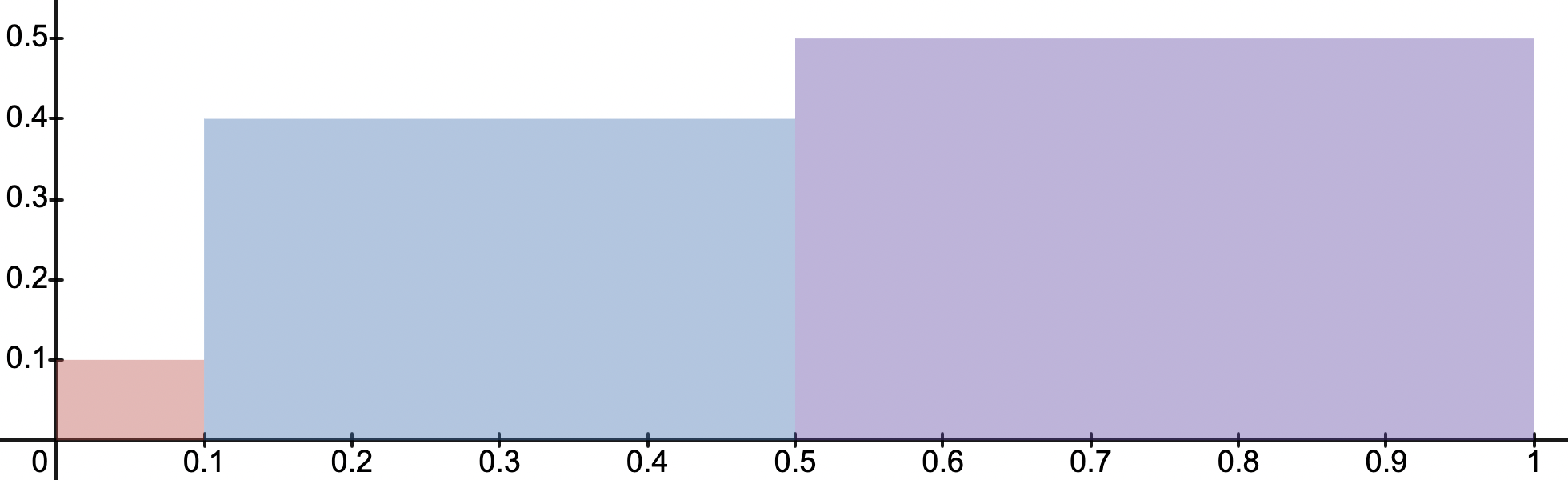}}\\ %
\subfigure[{Sketches of $p_1$ (red), $p_2$ ([0.6,0.2,0.2], blue), and their profile (black) }]{\label{fig:profile-C}\includegraphics[width=0.45\textwidth]{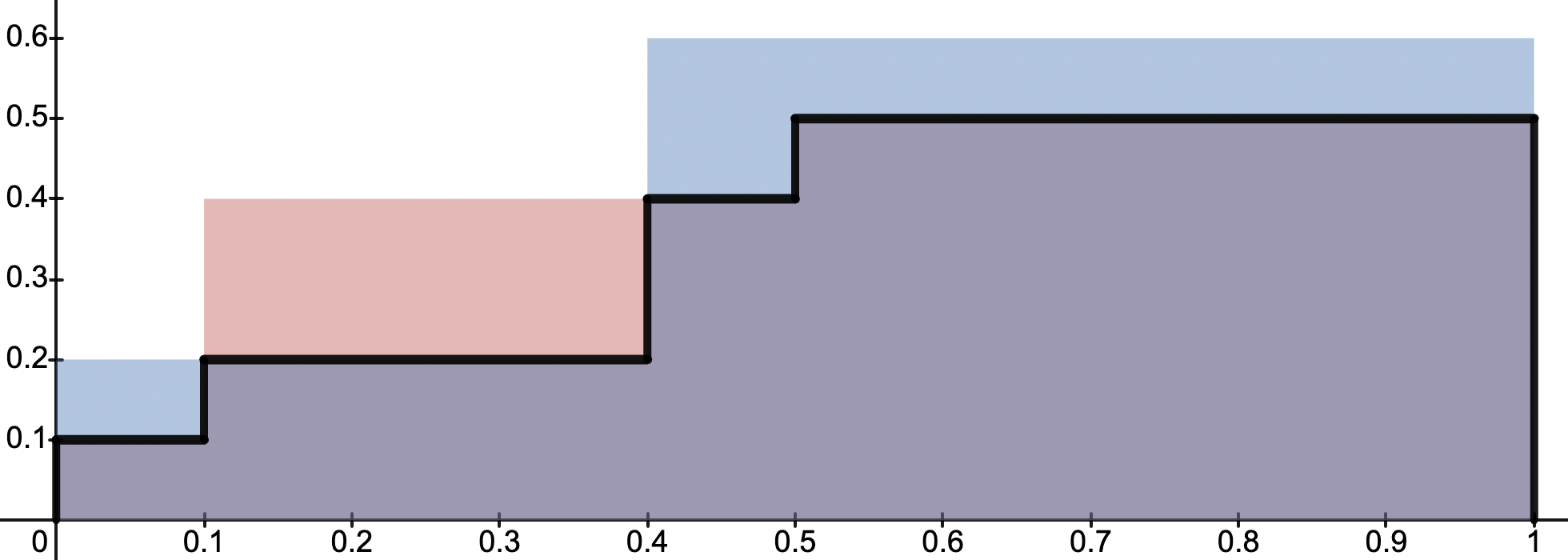}}\hspace{0.3in}
\subfigure[{Sketches of $p_1$ (red), $p_2$ (blue), their profile (black), and a valid coupling ([0.5,0.2,0.2,0.1], orange)}]{\label{fig:profile-D}\includegraphics[width=0.45\textwidth]{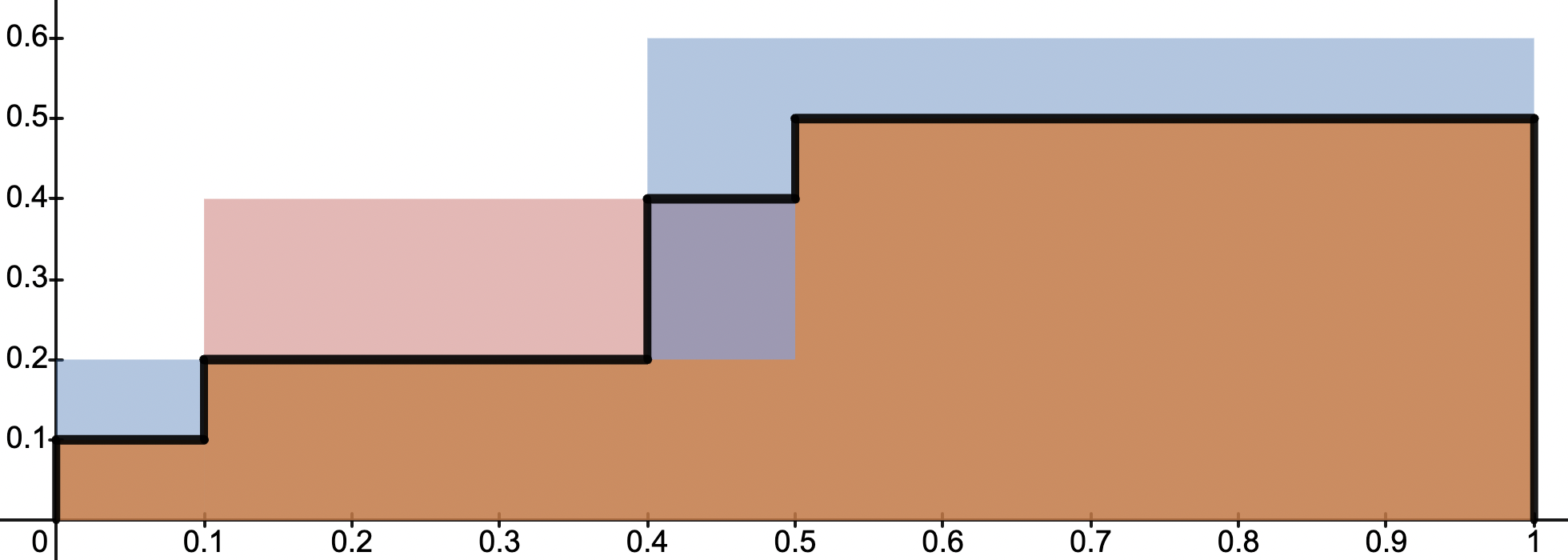}}%
\caption{Visualizations of the profile method.}
\label{fig:profile-pics}
\end{center}
\end{figure*}

\paragraph{Motivation.} In designing a new lower bound, we begin from first principles and examine what properties a valid coupling must obey. For example, consider a distribution $p_1 = [0.5, 0.4, 0.1]$. We would like to analyze necessary conditions for a valid coupling including $p_1$. Any valid coupling must have states that map to the smallest state of $p_1$ of size $0.1$. More concretely, these coupling states will have total size $0.1$ and each such state must be of size $\le 0.1$. As such, the coupling must have at least $0.1$ mass that comes from states of size $\le 0.1$. A similar statement is true from the coupling needing to have states that map to the state of $p_1$ of size 0.4: at least $0.1+0.4=0.5$ mass must come from states of size $\le 0.4$. There are accordingly three such necessary constraints for a valid coupling:
\begin{enumerate}
    \item At least $0.1$ mass that comes from states of size $\le 0.1$.
    \item At least $0.5$ mass that comes from states of size $\le 0.4$.
    \item At least $1.0$ mass that comes from states of size $\le 0.5$.
\end{enumerate}
In \Cref{fig:profile-A} we visualize these constraints with the former quantity corresponding to the $x$ dimension, and the latter quantity as the $y$ dimension. It is fitting that these constraints are left-aligned to maximize their overlap, as this corresponds to the constraints being minimally binding.  If we clean this visualization by eliminating parts of constraints that are not binding, we obtain \Cref{fig:profile-B}. Notably, this visualization provides a \emph{sketch} of the distribution $p_1$, where the sketch is a visualization in which the $i$-th state corresponds to a box of width and height $p_1(i)$. In terms of this visualization, \emph{our necessary conditions are exactly enforcing that the sketch of any valid coupling must not exceed the sketch of $p_1$}. We generalize this motivation by considering the corresponding constraints and visualization in the presence of another distribution $p_2 = [0.6, 0.2, 0.2]$. Combining the constraints imposed by coupling $p_1$ and $p_2$, they enforce that the sketch of any valid coupling must not exceed the lower-envelope of the input distributions' sketches. We will call this lower-envelope the \emph{profile} of the input distributions. In \Cref{fig:profile-C}, these constraints are visualized along with the profile of the sketches for $p_1$ and $p_2$. As an example, one valid coupling for $p_1$ and $p_2$ has the distribution $[0.5, 0.2, 0.2, 0.1]$. As visualized in \Cref{fig:profile-D}, the sketch of the coupling never exceeds the profile of $p_1$ and $p_2$, as is required by any valid coupling. 

More generally, the profile method will work to obtain a lower bound for the optimal coupling by relating this profile-related constraint on the sketch representation of any valid coupling to the required entropy for any valid coupling. In particular, we will observe how for a point on the profile of height $y$, this will correspond to some mass of the coupling that must come from a state of size $\le y$, and thus must be contributing information $\ge \log(\frac{1}{y})$. This will enable us to integrate over the profile after such a logarithmic transformation to obtain a lower bound for the optimal coupling. Note that the existing lower bound of $\andd S$ does not utilize the constraints we outline in this motivation, and it is by using the information provided by these constraints that the profile method is able to obtain a stronger bound.

\paragraph{The profile method.}
We now more concretely describe the profile method, starting with its sketch visualization. For a (possibly partial) distribution $p$, consider visualizing it in a way similar to a sketch, where the states are sorted from left to right in non-decreasing order of size, and a state is represented by a square box with width and height $p(i)$. We formally define it as follows:

\begin{restatable}[]{definition}{sketchdef}
\label{def:sketch}
$\sketch_p(x)$ is defined for $x \in (0,\mass(p)]$ such that $\sketch_p(x) = p(i)$ for $x \in (\sum_{j>i}p(j), \sum_{j\ge i}p(j)]$.
\end{restatable}

For a particular example where $S$ contains two distributions $p_1 = [0.6,0.2,0.2]$ and $p_2 = [0.5,0.4,0.1]$, we illustrate their sketches in \Cref{fig:profile-C}. As previously explained, the constraints of coupling necessitate that the sketch of any valid coupling does not exceed the sketch of any input distribution. Another way in which sketches are insightful is that we can relate the entropy of a distribution to its sketch. As a point on a distribution's sketch of height $y$ corresponds to mass contributing information $\lg(\frac{1}{y})$, the entropy of a distribution is exactly obtained by integrating over its sketch after applying a logarithmic transformation:
\begin{remark} \label{remark:integral-p}
$H(p)=\int_0^{\mass(p)} \lg \left( \frac{1}{\sketch_p(x)} \right) \, dx$.
\end{remark}

Recall how the sketch of any valid coupling must not exceed the sketch of any input distribution. Equivalently, we can say the sketch of any valid coupling must not exceed the lower-envelope of the input distributions' sketches. We refer to this lower-envelope as \emph{the profile}:
\begin{restatable}[]{definition}{profiledef}
\label{def:profile}
$\profile_S(x) \triangleq \min\limits_{p \in S} \sketch_p(x)$.
\end{restatable}
An example of the profile is illustrated in \Cref{fig:profile-C}. Note how (unlike $\andd S$) the profile does not necessarily correspond to a partial distribution nor does it necessarily consist of squares. Intuitively, it is a much more flexible object that will allow us to obtain a better lower bound for the optimal coupling. Similar to our observation in \cref{remark:integral-p} that the integral of a distribution's sketch after a log-transformation is equal to its entropy, we define an analogous quantity for the profile:
\begin{restatable}[]{definition}{profileintdef}
\label{definition:profile-int} Let the profile ``entropy'' be
$H(\profile_S) \triangleq \int_0^{\mass(S)} \lg \left( \frac{1}{\profile_S(x)} \right)  dx$.
\end{restatable}
As the sketch of any valid coupling must not exceed the profile, every point on the profile of height $y$ can be bijectively mapped to mass in the valid coupling from a state of size $\le y$, and thus contributing $\ge \lg(\frac{1}{y})$ entropy. This finally implies our novel lower bound for the optimal minimum entropy coupling (full proof in \cref{section:profile-lb-proof}):
\begin{restatable}[]{theorem}{profilelb} \label{theorem:profile-lb}
$H(\opt) \ge H(\profile_S)$.
\end{restatable}

\section{Better Greedy Coupling Approximation Guarantees}\label{section:coupling-guarantees}
In this section, we will use this new lower bound gained from the profile method to obtain better approximation guarantees for the greedy coupling algorithm. 
\subsection{Coupling Two Distributions} \label{section:two-couple}

For coupling just $m=2$ distributions, the surprising state of affairs is that there are three distinct algorithms that all attain a 1 bit additive approximation guarantee as shown in \cite{cicalese2019minimum,rossi2019greedy,li2021efficient}. It is then natural to wonder if this is the best possible approximation guarantee. Further, we can construct instances where $H(\opt) \approx H(\andd S) + 0.66$ (see \cref{section:gaps}), meaning it is impossible to get a better approximation guarantee than 0.66 bits with respect to the previously used lower bound of $H(\andd S)$. Yet, with a proof relating to the profile method lower bound, we break the majorization barrier and show that the greedy coupling algorithm of \cite{kocaoglu2017entropic} is always additively within $\approx 0.53$ bits of the optimum:

\begin{restatable}[]{theorem}{twoCouple}
\label{thm:0.53}

For $m=2$, $H(\G_S)\le H(\profile_S)+\frac{\lg e}{e} \approx H(\profile_S)+0.53$.

\end{restatable}

\paragraph{Proof intuition for \cref{thm:0.53}.}
Let $S^t$ denote the distributions of $S$ after the $t$-th state of the greedy coupling. Similarly, $\greed^t$ represents the vector $\greed(1),\dots,\greed(t)$. We will consider the evolution of the non-decreasing monovariant $M^t\triangleq H(\profile_{S^t})+H(\greed^t)$. Initially, nothing is coupled and we have $M^0=H(\profile_{S^0}) + H(\greed^t) = H(\profile_S) + 0$. In the end, everything is coupled and we have $M^{|\greed|}=H(\profile_{S^{|\greed|}}) + H(\greed^{|\greed|}) = 0 + H(\greed)$. We will show that the profile method implies 
\begin{equation}
M^{t+1}\le M^{t}+\frac{\lg e}{e}\greed(t+1).\label{ineq:m=2-monovariant}
\end{equation}
\cref{ineq:m=2-monovariant} would suffice to prove the conclusion because then
\begin{align*}
H(\greed)=M^{|\greed |}
&\le M^0 + \frac{\lg e}{e}\sum_{t=1}^{|\greed|}\greed(t)\\
&= H(\profile_S) + \frac{\lg e}{e}.
\end{align*}
We now outline the derivation of \cref{ineq:m=2-monovariant}. Consider the two distributions in $S^t$ to be denoted by $p^t$ and $q^t$ (sorted to be in non-increasing order). Without loss of generality, suppose $p^t(1) \le q^t(1)$. This means our greedy algorithm will select $\greed(t+1) = p^t(1)$, the state $p^t(1)$ will be fully coupled, and $q^t(1)$ will afterwards have size $q^t(1)-p^t(1)$. With some calculation, we can bound the increase in our monovariant $M^{t+1}-M^{t}$ by $(q^t(1)-p^t(1)) \cdot \max\paren{0,\lg(\frac{1}{q^t(1)-p^t(1)}) - \lg(\frac{1}{p^t(1)})}$.  We can in turn bound this expression by
\begin{align}
& p^t(1)\cdot\!\!\! \max_{p^t(1)< q^t(1)< 2p^t(1)}\!\left [\!\frac{q^t(1)-p^t(1)}{p^t(1)}\lg\!\paren{\!\frac{p^t(1)}{q^t(1)-p^t(1)}\!}\!\right ]\nonumber\\
& = p^t(1)\cdot \max_{0< r< 1}\left[r\lg \frac{1}{r}\right] = p^t(1)\cdot \frac{\lg e}{e}, \label{eq:final-0.53}
\end{align}

where the maximum is attained at $r=e^{-1}$. As $\greed(t+1) = p^t(1)$, \cref{eq:final-0.53} implies \cref{ineq:m=2-monovariant}, thus concluding our proof outline. Hence, the profile method enables us to naturally analyze how the monovariant of the profile and greedy evolves, and attain an approximation guarantee breaking the majorization barrier. The full proof is in \cref{section:0.53-proof}.

\textit{Remark.} This proof method can also show guarantees for small $m$ that are better than the guarantee for general $m$ in \cref{thm:1.22}. More details are included in \cref{section:small-m}.

\subsection{Coupling an Arbitrary Number of Distributions} \label{section:many-couple}

For coupling many distributions, the result of \cite{compton2022tighter} showed that the greedy coupling algorithm of \cite{kocaoglu2017entropic} is within $\approx 1.44$ bits of the optimal and that it is impossible to get a better guarantee with respect to the $H(\andd S)$ lower bound. Here we use the profile lower bound to break the majorization barrier and show the greedy coupling is always within $\approx 1.22$ bits of the optimal:

\begin{restatable}[]{theorem}{manyCouple}
\label{thm:1.22}

$H(\G_S)\le H(\profile_S)+\frac{1+\lg e}{2} \approx H(\profile_S)+1.22$.
\end{restatable}

\paragraph{Proof intuition for \cref{thm:1.22}.} For this proof, we will find it helpful to upper bound the total remaining mass to be coupled as a function of the size of the greedy's next state. More concretely, we will define a function $\rmass_S(y)$ that corresponds to an upper bound on the total remaining mass to be coupled if the greedy coupling is about to create a state of size $\le y$. For any non-decreasing function $\rmass_S$ satisfying this condition, the following quantity upper bounds the cost of the greedy coupling solution:
\begin{equation} \label{equation:rmass-greed-bound}
    H(\greed) \!\le\! \int_0^1 \!\rmass'_S(y) \!\cdot\! \lg\left(\frac{1}{y}\right)  dy.
\end{equation}
This holds because, for any $y$, the increase in our remaining mass upper bound denoted by $\rmass'_S(y)$ corresponds to mass that must come from greedy states of size $\ge y$. Adversarially, we say this mass contributes information $\lg(1/y)$. It remains to be answered how one can construct a $\rmass_S$ that upper bounds the remaining mass given the next greedy state is of size $\le y$. Our most intuitive bound is:
\begin{align} \label{equation:simple-rem-mass}
&\rmasssim_{S}(y) \triangleq  \\\nonumber
&\max\limits_{p \in S}\rmasssim_{p}(y) \triangleq \max_{p\in S} \sum_{j=1}^n \min(p(j),y).
\end{align}
We design the bound $\rmasssim_{p'}(y)$ for a particular distribution $p'$ such that if the largest remaining state in $p'$ is $\le y$, then $\rmasssim_{p'}(y)$ will be a valid upper bound on the total remaining mass in $p'$. Recall that there must be such a $p'$, given how the greedy is defined in \cref{alg:greedy}. So, consider a $p'$ where its largest remaining state has size $\le y$. Then the amount of mass remaining in any particular state $j$ of $p'$ is $\le \min(p'(j),y)$, hence the bound given in \cref{equation:simple-rem-mass}. As an example, if the largest remaining state of $p_2$ (specified in \Cref{fig:profile-C}) is at most $0.3$, then the remaining mass in $p_2$ is at most $\rmasssim_{p_2}(0.3) = \min(0.5, 0.3) + \min(0.4,0.3) + \min(0.1,0.3) = 0.7$. Through further non-trivial analysis, one can relate the bound of \cref{equation:rmass-greed-bound} when using $\rmasssim$ defined in \cref{equation:simple-rem-mass} to show that $H(\opt) \le H(\profile_S) + \lg(e) \approx H(\profile_S) + 1.44$, matching the guarantee of \cite{compton2022tighter}.

However, we can improve upon $\rmasssim$. In our example with $p_2$, we stated that if the largest remaining state of $p_2$ is at most $0.3$, then the remaining mass in its state $p_2(1)$ was $\le \min(0.5,0.3) = 0.3$. While this is true, it is not tight. If the largest remaining state in $p_2$ is $r_{\textrm{max}}$ where $r_{\textrm{max}} < 0.5$, then given the monotonicity in the size of greedy states, there must have been a greedy state that removed at least $r_{\textrm{max}}$ mass from $0.5$. More formally, we can bound that the remaining mass in $0.5$ is at most $0.5 - r_{\textrm{max}}$. Combining this with the prior upper bound that its remaining mass is $\le r_{\textrm{max}}$, we can show how the remaining mass in $p_2(1)$ is at most $\max_{0 \le r_{\textrm{max}} \le 0.3} \min(r_{\textrm{max}},0.5-r_{\textrm{max}}) = 0.25$. We can generalize this reasoning as follows:
\begin{align} \label{equation:advanced-rem-mass}
& \rmassadv_{S}(y) \nonumber \\
& \triangleq \max\limits_{p \in S}\rmassadv_{p}(y) \nonumber \\ 
& \triangleq \max\limits_{p \in S}\sum_{j=1}^n \begin{cases}
y & y\le \frac{p(j)}{2} \\
p(j)/2 & \frac{p(j)}{2}<y<p(j) \\
p(j) & p(j) \le y \\
\end{cases}
\end{align}
In our example, this attains $\rmassadv_{p_2}(0.3) = 0.25 + 0.2 + 0.1 = 0.55$, improving upon $\rmasssim_{p_2}(0.3) = 0.7$. Through further non-trivial analysis, one can relate the bound of \cref{equation:rmass-greed-bound} when using $\rmassadv$ defined in \cref{equation:advanced-rem-mass} to show that $H(\opt) \le H(\profile_S) + \frac{1 + \lg(e)}{2} \approx H(\profile_S) + 1.22$, attaining a new best guarantee for coupling many distributions and breaking the majorization barrier. We defer the full proof to \cref{section:proof-1.22}.

\section{Computing Exactly Optimal Solutions}\label{section:exactly-optimal}

While earlier sections discuss approximating the minimum entropy coupling, this section focuses on approaches to compute an \emph{exact} solution. This is known to be generally intractable as it is NP-Hard. It was posed as an open problem in \cite{kovavcevic2015entropy} whether the problem is in NP. In viewing the minimum entropy coupling as an optimization problem with linear constraints, there are exponentially many ($n^m$) variables. Perhaps surprisingly, we show there is always an optimal coupling that uses relatively few states:

\begin{restatable}[]{lemma}{smallsuplemma}
\label{lemma:small-sup}

There always exists an optimal solution to \coupleProb \space with support size of at most $nm - (m-1)$.

\end{restatable}

\begin{proof}
We can formulate minimum entropy coupling as an optimization problem over the $d\triangleq n^m$-dimensional polytope $\BP(i_1,\dots, i_m)$:
\begin{equation*}
\begin{aligned}
& \underset{\BP}{\text{minimize}}
& & \sum_{i=(i_1,\dots, i_m)}\BP(i)\lg(1/\BP(i)) \\
& \text{subject to}
& & \BP(i)\ge 0, \forall i ;\, \sum_i \BP(i)=1 \\
&&& \sum_{i:i_j=k}\BP(i)=\mathbf \p_j(k), \forall 1\le j\le m, 1\le k < n.
\end{aligned}
\end{equation*}
Observe that there are a total of $q\triangleq m(n-1)+1=nm-(m-1)$ equality constraints. We have omitted the constraints for $\p_j(k=n)$ because they are linear combinations of the constraint $\sum_i \BP(i)=1$ and the constraints for $\p_j(k<n)$.

Since the objective is concave, there exists a vertex of $\BP$ at which the objective is minimized. To finish, it suffices to show every vertex of this polytope has at most $q$ nonzero entries. This is true because every vertex is the unique point in the intersection of some $d$ constraint hyperplanes. At least $d-q$ of these hyperplanes must be of the form $\BP(i)=0$, so every vertex must have at least this number of zeros.
\end{proof}

This bound of $nm - (m-1)$ is particularly significant, as we know the number of states the greedy coupling algorithm uses is upper bounded by the same value (shown in \cite{kocaoglu2017entropic}). More generally, it is notable that there is always an optimal solution with support size of at most $nm - (m-1)$, because for almost all instances (in the measure 1 sense) \emph{there is no possible coupling that uses fewer than $nm - (m-1)$ states}. Furthermore, this enables an exponential-time algorithm (simply enumerating over all polytope vertices). This also makes progress on the open problem of \cite{kovavcevic2015entropy} regarding the NP membership of minimum entropy coupling (MEC), as one could use the constraints identifying the polytope vertex as a certificate. Given the polytope vertex, all that remains is the seemingly trivial task of computing the Shannon entropy of a discrete distribution:

\begin{restatable}[]{corollary}{npreduct}
\label{theorem:np-reduct}
$\coupleProb \le_{\textrm{NP}} \entropyProb$.
\end{restatable}

For reasons involving complexity of finite precision arithmetic (e.g. \cite{kayal2012sum}), it is unknown whether $\entropyProb$ is in NP, despite the fact in practice this appears to be a simple task. We note that \cref{theorem:np-reduct} assumes all inputs are represented as rational numbers.

\subsection{Speeding Up Backtracking} \label{section:experiments}

Recent works such as \cite{sokota2022communicating,anonymous2023perfectly} have applied minimum entropy couplings for $m=2$ distributions in settings where finding better couplings has concrete benefits. It is thus natural to wonder how quickly we can exactly couple two distributions. As enumerating over the exponentially many polytope vertices is computationally expensive, we propose other algorithms that can more efficiently couple two distributions (although they are still exponential-time algorithms). In \cref{app:dp}, we introduce a new dynamic programming algorithm that runs in $O(9^n\cdot \text{poly}(n))$ time. This approach utilizes a new characterization of couplings as spanning trees over the distributions, similar to the perspective of \cref{lemma:monotone-backtracking}. In this section, we present a backtracking algorithm that recursively chooses the next state of the coupling and eliminates the smallest size marginal the state maps to. This is a generalization of how the greedy coupling algorithm always eliminates the state corresponding to the minimum over all remaining distributions' maximums. Dynamic programming and backtracking can be seen as two ways of leveraging our new spanning tree perspective. While the dynamic programming approach is relatively final, the backtracking approach is crucially improved by new lower bounds for coupling as they help prune its search space, meaning \emph{better lower bounds directly improve its speed}.

Our motivation is to enable the computation of higher-quality couplings than the polynomial-time greedy coupling in settings where enumerating over all polytope vertices is computationally infeasible. Backtracking is exactly optimal for $m=2$. Further details are provided in \cref{app:back}. 

We use the following lemma to prove the validity of our backtracking algorithm for $m=2$.

\begin{restatable}[]{lemma}{monotoneback}
\label{lemma:monotone-backtracking}

Consider the weighted undirected bipartite graph $G$ with vertex set $\{\ell_{1\dots n}, r_{1\dots n}\}$ and edge set $\{(\ell_i,r_j,\p(i,j))) \mid \p(i,j) > 0 \}$ induced by some minimum-entropy coupling $\p$. This graph is a forest, and any edge with the maximum weight must have a leaf of the forest as one of its endpoints. 

\end{restatable}
\begin{table*}[htbp] 
\centering
\caption{Average Runtime (in Seconds) for Exactly Coupling Two Distributions} \label{table:experiment}
\scriptsize
\begin{tabular}{|c|c|c|c|c|c|c|c|c|}
\hline
\textbf{}&\multicolumn{8}{|c|}{\pmb{$n$}} \\
\cline{2-9}
\textbf{Algorithm}& 4 & 5 & 6 & 7 & 8 & 9 & 10 & 11\\
\hline
{Naive Polytope Vertex Enumeration} & $0.002$ & $0.189$ & $33.79$ & $>120$ & $>120$ & $>120$ & $>120$ & $>120$ \\
\hline
{Backtracking [$0$]} & $0.0003$ & $0.004$ & $0.106$ & $3.576$ & $>120$ & $>120$ & $>120$ & $>120$ \\
\hline
{Backtracking [$H(\andd S)]$} & $0.0001$ & $0.001$ & $0.013$ & $0.256$ & $4.907$ & $>120$ & $>120$ & $>120$ \\
\hline
{Backtracking [$H(\profile_S)$]} & $0.0001$ & $0.0009$ & $0.009$ & $0.153$ & $2.206$ & $>120$ & $>120$ & $>120$ \\
\hline
{Backtracking [$H(\majorpro_S)$]} & $0.00006$ & $0.0004$ & $0.003$ & $0.035$ & $0.344$ & $5.398$ & $>120$ & $>120$ \\
\hline
{Dynamic Programming} & $0.00007$ & $0.0004$ & $0.002$ & $0.012$ & $0.093$ & $0.802$ & $9.769$ & $>120$ \\
\hline
\end{tabular}
\label{tab2}
\end{table*}
Our backtracking algorithm is guaranteed to find any coupling satisfying the lemma above, because it can repeatedly peel away the maximum weight edge of the forest. The proof of \cref{lemma:monotone-backtracking} is deferred to \cref{section:exactly-optimal-proof}. 

\paragraph{Combining lower bounds.} We now discuss a better lower bound we use to improve the speed of backtracking. Recall that the intuitions behind the lower bounds $H(\andd S)$ and $H(\profile_S)$ are somewhat different. Interestingly, we are able to combine them to get an even better lower bound than simply taking $\max (H(\andd S), H(\profile_S) )$: %
\begin{restatable}[]{definition}{majorprodef}
\label{def:majorpro}
$\majorpro_S \triangleq \argmin\limits_{d} H(d)$ $ \, s.t. \, \sketch_d(x) \le \profile_S(x) \, \forall x \in (0,\mass(S)]$
\end{restatable}
\begin{restatable}[]{theorem}{majorprothm}
\label{theorem:major-pro}
$H(\opt) \ge H(\majorpro_S) \ge H(\profile_S), H(\andd S)$
\end{restatable}

In words, $\majorpro_S$ is the minimum entropy distribution whose sketch never exceeds the profile (as is shown to be a necessary condition for a valid coupling in the proof of \cref{theorem:profile-lb}). On the surface, this seems only related to the $H(\profile_S)$ lower bound and not $H(\andd S)$. Yet, one can show that this constraint actually enforces $\majorpro_S \majorby \andd S$. 
This means we have a lower bound that is always at least as good as both $H(\andd S)$ and $H(\profile_S)$. The full proof of \cref{theorem:major-pro} is deferred to \cref{section:major-pro-proof}. 

Note that $H(\profile_S)$ and $H(\majorpro_S)$ can both be efficiently computed in near-linear time,  as will be leveraged by our backtracking algorithm. Details about their computation are deferred to \cref{app:lower-algs}.

\paragraph{Experimental results.} In \cref{table:experiment} we show experiments comparing the average running time of coupling two distributions with different (provably optimal) algorithms. All algorithms we present have stronger performance than the naive baseline of enumerating over all polytope vertices. We also observe clear improvements in the speed of backtracking with better lower bounds than the previously best-known lower bound of $H(\andd S)$. It is particularly noteworthy that $\majorpro_S$ provides better performance than $\profile_S$, as our theoretical approximation guarantees in \cref{section:two-couple,section:many-couple} do not immediately yield any benefit for using $\majorpro_S$ in place of $\profile_S$. This seems to indicate that $\majorpro_S$ is a meaningfully stronger lower bound for many instances, and may be a helpful reference for future work on minimum entropy coupling approximation guarantees. We observe the best performance from our dynamic programming algorithm and backtracking algorithm with $\majorpro_S$. Due to the nature of its pruning, the backtracking has higher variance. Finally, we emphasize that any future work providing new lower bounds for minimum entropy coupling could immediately improve our backtracking method. Additional experiment details are provided in \cref{section:experiment-details}.

\section{Conclusion}
In this work, we showed advances in algorithms, lower bounds, and approximation guarantees for the minimum entropy coupling problem. 
In \cref{section:concave-costs}, we discuss how our methods generalize to any concave minimum-cost coupling. 
There are many avenues for future work, such as examining how previously discussed applications of minimum entropy coupling may concretely benefit from our results. 
For example, the algorithms for entropic causal inference used in \cite{kocaoglu2017entropic,compton2020entropic,compton2022entropic} do not use the construction of the coupling, just the value of its entropy. As our new lower bounds on the value are more efficient to compute than the greedy coupling and have the same theoretical guarantees, using these instead would directly improve the speed of these causal discovery algorithms. 
There are many remaining theoretical questions. Hardness of approximation remains an open problem. It is still unknown the best approximation guarantee that greedy coupling achieves, or whether there exists a polynomial-time algorithm with strictly better worst-case approximation guarantees than greedy coupling. %
Deriving an analysis that more smoothly interpolates best-known guarantees from $m=2$ to $m=\infty$ is also an open question. Another open question is whether one can obtain a better approximation guarantee with respect to $\majorpro_S$ than with $\profile_S$. Finally, in \cref{section:gaps} we provide some discussion on gaps between various relevant quantities, with counter-examples that were computationally discovered by local search. For example, we provide instances where the algorithms of \cite{cicalese2019minimum} and \cite{li2021efficient} have optimality gaps greater than $\log(e)/e \approx 0.53$ for coupling two distributions, whereas \cref{thm:0.53} shows this is impossible for greedy coupling. It is our hope that this information will help guide algorithm selection and %
be insightful for future conjectures.

\subsubsection*{Acknowledgements}

This work was supported in part by the National Defense Science \& Engineering Graduate (NDSEG) Fellowship Program, NSF Grant CAREER 2239375, and the MIT-IBM Watson AI Lab.

\bibliography{paperbib}
\bibliographystyle{apalike}

\flushcolsend

\clearpage

\onecolumn
 \vbox{%
    \hsize\textwidth
    \linewidth\hsize
    \vskip 0.1in
    \hrule height 4pt
  \vskip 0.25in
  \vskip -\parskip%
    \centering
    {\LARGE\bf Appendix \par}
    \vskip 0.29in
  \vskip -\parskip
  \hrule height 1pt
  \vskip 0.09in%
  }

\appendix
\section{Proofs for Section \ref{section:novel-lower-bound}: the Profile Method Lower Bound}

\subsection{Proof of \cref{theorem:profile-lb}} \label{section:profile-lb-proof}

\profilelb*

While in our proof outline in the main text we discussed \textbf{sketches} and \textbf{profiles}, our formal proof in this section is simplified by analyzing their ``inverses,'' which we define next.

\begin{definition}
For a (possibly partial) distribution $p_i$, define $\isketch_{p_i}(y)$ to be the non-decreasing function $\isketch_{p_i}\colon [0,1]\to [0,1]$ such that $\isketch_{p_i}(y)$ equals the sum of all probability masses of $p_i$ less than or equal to $y$. That is,
\begin{equation*}
\isketch_{p_i}(y)\triangleq \sum_{j=1}^{|p|}p_i(j)\cdot [p_i(j)\le y].
\end{equation*}
For a set of probability distributions $S$, we define $\iprof_S(y)$ as
\begin{equation*}
\iprof_S(y)\triangleq \max_{p_i\in S}\isketch_{p_i}(x). 
\end{equation*}
\end{definition}

We refer to $\isketch_p$ and $\iprof_S$ as ``inverse sketches'' and ``inverse profiles,'' respectively. Note that for any $S$, the graph of $\iprof_S$ may be obtained by transposing the graph of $\profile_S$, though they are not strictly inverses because neither function is strictly increasing.

We also extend the definition of entropy to inverse profiles.

\begin{definition}\label{def:inverse-entropy}

For an inverse profile (or inverse sketch) $h$, we define the entropy of $h$ to be:
\begin{equation}\label{eq:entropy}
H(h)\triangleq \int_{0}^1\frac{h(y)}{y \ln 2}\, dy.
\end{equation}

\end{definition}

\begin{lemma}\label{lemma:entropy-consistent}

The definition of entropy is consistent for profiles and inverse profiles (as well as sketches and inverse sketches); that is, $H(\profile_S)=H(\iprof_S)$.

\end{lemma}

\begin{proof}
Start with \Cref{definition:profile-int}. First, we substitute $x$ in place of $y$:
\begin{align*}
\profile_S(x)&=y,\\
x &= \iprof_S(y), \\
dx &=\iprof_S'(y)\,dy, \text{\footnotemark}
\end{align*}
which implies that
\begin{align*}
\int_0^{\mass(S)}\log\paren{\frac{1}{\profile_S(x)}}\, dx
&=\frac{1}{\ln 2}\int_0^{\mass(S)}\ln\paren{\frac{1}{\profile_S(x)}}\, dx,\\
&=\frac{1}{\ln 2}\int_0^1 \iprof_S'(y) \ln\paren{\frac{1}{y}}\, dy.
\end{align*}
Next we integrate by parts:
\begin{align*}
&=\frac{1}{\ln 2}\paren{\iprof_S(y) \ln\paren{\frac{1}{y}}\Big |_0^1+\int_0^1 \frac{\iprof_S(y)}{y}},\\
&= \frac{1}{\ln 2}\int_0^1 \frac{\iprof_S(y)}{y},
\end{align*}
as desired.
\end{proof}

\footnotetext{Technically, $\iprof_S'$ is not defined at the points where $\iprof_S$ is discontinuous. We can circumvent this issue by letting $\iprof_S'$ take on the value of the appropriate multiple of the Dirac delta function at such points.}

\begin{lemma}\label{lemma:hist-lb}

$\isketch_{OPT_S}(y)\ge \iprof_S(y)$ for all $y$. 

\end{lemma}

\begin{proof}

By definition of $\iprof_S$ it suffices to show that $\isketch_{\opt}(y)\ge \isketch_p(y)$ for all $p\in S$. Note that $\opt$ may be derived from $p$ by repeatedly performing the following operation: take some state in $p$ and split it into two smaller states, transforming $p$ into $p'$. 

Suppose that an operation splits a state of mass $m_1+m_2$ into two states of masses $m_1\le m_2$. Then it is easy to see that
\begin{itemize}
    \item $\isketch_{p'}(y)=\isketch_p(y)$ for $y<m_1$
    \item $\isketch_{p'}(y)>\isketch_p(y)$ for $m_1\le y<m_1+m_2$
    \item $\isketch_{p'}(y)=\isketch_p(y)$ for $m_1+m_2\le y$.
\end{itemize}
It follows that if $p'$ can be derived from $p$ via some sequence of operations, $\isketch_{p'}(y)\ge \isketch_p(y)$ for all $y$. Substituting $\opt$ in place of $p'$ yields the desired result.

\end{proof}

\paragraph{Proof of \cref{theorem:profile-lb}} We can now prove \cref{theorem:profile-lb}.
By consistency of our definitions of entropy, $H(\opt)=H(\isketch_{\opt})$ and $H(\profile_S)=H(\iprof_S)$. 
To finish, $H(\isketch_{\opt})\ge H(\iprof_S)$ follows by combining \cref{lemma:hist-lb} with \Cref{def:inverse-entropy}.

\section{Proof of Theorem \ref{thm:0.53}: Better Greedy Coupling Approximation Guarantee for Two Distributions} \label{section:0.53-proof}

\twoCouple*

It suffices to derive \cref{ineq:m=2-monovariant} in the main text. As in \cref{section:profile-lb-proof}, we find it more convenient to prove our results in terms of $H(\iprof_S)$ instead of $H(\profile_S)$ (which is allowed by \cref{lemma:entropy-consistent} in \cref{section:profile-lb-proof}).

After the greedy algorithm selects $\greed(t+1)=p^t(1)$, we have
\begin{itemize}
    \item $H(\greed^{t+1})=H(\greed^t)+p^t(1)\lg(1/p^t(1))$.
    
    \item $\isketch_{p^{t+1}}(x)=\begin{cases}
    \isketch_{p^{t}}(y) & y < p^t(1) \\
    \isketch_{p^{t}}(y)-p^t(1) & p^t(1) \le y \\
    \end{cases}$
    
    \item $\isketch_{q^{t+1}}(y)=\begin{cases}
    \isketch_{q^{t}}(y) & y < q^t(1)-p^t(1) \\
    \isketch_{q^{t}}(y)+q^t(1)-p^t(1) & q^t(1)-p^t(1) \le y < q^t(1) \\
    \isketch_{q^{t}}(y)-p^t(1) & q^t(1)\le y \\
    \end{cases}$
\end{itemize}

From the above equations, we conclude that:

\begin{equation}
\iprof_{S^{t+1}}(y)\:\begin{cases}
     =\iprof_{S^t}(y) & y < q^t(1)-p^t(1) \\
     \le \iprof_{S^t}(y)+q^t(1)-p^t(1) & q^t(1)-p^t(1)\le y < p^t(1) \\
     =\iprof_{S^t}(y)-p^t(1)=\mass(S^{t+1}) & p^t(1)\le y \\
    \end{cases},\label{eq:coupling-2}
\end{equation}

which in turn implies

\begin{align*}
H(\iprof_{S^{t+1}})&\le H(\iprof_{S^t})-p^t(1)\lg(1/p^t(1))+\begin{cases}
(q^t(1)-p^t(1))\int_{q^t(1)-p^t(1)}^{p^t(1)}\frac{1}{y\ln 2}\,dy & q^t(1)<2p^t(1) \\
0 & \text{otherwise}
\end{cases}\\
&\le H(\iprof_{S^t})-p^t(1)\lg(1/p^t(1))+\begin{cases}
\frac{q^t(1)-p^t(1)}{\ln 2}\cdot \ln \frac{p^t(1)}{q^t(1)-p^t(1)} & q^t(1)<2p^t(1) \\
0 & \text{otherwise}
\end{cases}.
\end{align*}

Combining this inequality with our expression for $H(\greed^{t+1})$ gives 
\begin{align}
M^{t+1}&\le M^t+p^t(1)\cdot \max_{p<q<2p}\frac{1}{\ln 2}\cdot \frac{q-p^t(1)}{p^t(1)}\ln \frac{p^t(1)}{q-p^t(1)}\nonumber\\
&= M^t+\greed(t+1)\cdot \frac{1}{\ln 2}\cdot \max_{0<r<1}\left[r\ln (1/r)\right].\label{align:almost-done}
\end{align}
As the function in brackets is concave, we can attempt to maximize it by setting its derivative to 0:
\begin{equation*}
0=\frac{d}{dr}(r\ln(1/r))=\ln(1/r)-1\implies r=e^{-1}.
\end{equation*}
Since $e^{-1} \in (0,1)$, this maximizes the function in brackets in the specified range.
Plugging this value of $r$ into \cref{align:almost-done} gives 
\begin{equation*}
M^{t+1}\le M^t+\greed(t+1)\cdot \frac{1}{e\ln 2}=M^t+\greed(t+1)\cdot \frac{\lg e}{e},
\end{equation*}

as desired.

\section{Proof of Theorem \ref{thm:1.22}: Better Greedy Coupling Approximation Guarantee for Arbitrarily Many Distributions} \label{section:proof-1.22}

\manyCouple*

We prove this theorem using the following three lemmas. The first step is to upper bound $H(\greed)$ by an integral.

\begin{lemma}\label{lemma:integral}

Let $C_S(x)\colon [0,1]\to [0,1]$ be any non-decreasing function such that if there is $x$ mass left while running the greedy algorithm on $S$, the next state created by the greedy algorithm has mass at least $C_S(x)$. Then $H(\greed)\le \int_0^1 \lg \frac{1}{C_S(x)} \, dx$.

\end{lemma}

\begin{proof}
First rewrite $H(\greed)=H(\sketch_{\greed})$. Next, we show that $C_S$ lower bounds $\sketch_{\greed}$. By definition of $C_S$, $C_S(x)\le \sketch_{\greed}(x)$ whenever $x=1-\sum_{j=1}^i\greed(j)$. Since $C_S$ is non-decreasing and $\sketch_{\greed}(x)$ is constant on the interval $(1-\sum_{j=1}^{i+1}\greed(j), 1-\sum_{j=1}^i\greed(j)]$, it follows that $C_S(x)\le \sketch_{\greed}(x)$ for all $x\in (0,1]$. The conclusion now follows from recalling the definition of $H(\sketch_{G_S})$.

\end{proof}

We next aim to design a $C_S$. First, recall our definition of $\rmassadv_{S}(y)$.

\begin{align} \label{equation:advanced-rem-mass-copy}
& \rmassadv_{S}(y) \nonumber \\
& \triangleq \max\limits_{p \in S}\rmassadv_{p}(y) \nonumber \\ 
& \triangleq \max\limits_{p \in S}\sum_{j=1}^n \begin{cases}
y & y\le \frac{p(j)}{2} \nonumber \\
p(j)/2 & \frac{p(j)}{2}<y<p(j) \nonumber \\
p(j) & p(j) \le y \nonumber \\
\end{cases}
\end{align}

Informally, recall that $\rmass_{S}(y)$ is designed to correspond to an upper bound on the total remaining mass to be coupled if the greedy coupling is about to create a state of size $\le y$. Now, we describe our choice of $C_S$.

\begin{lemma}\label{lemma:define-C}
Define $C_S(x)\triangleq \argmin_y\left[\rmassadv_S(y)\ge x\right].$ Then $C_S$ satisfies the preconditions of \Cref{lemma:integral}.
\end{lemma}

\begin{proof}
Suppose that there is $x$ mass remaining while running the greedy algorithm, and then the greedy algorithm allocates a state of mass $y$. Then there exists $p_{i'}\in S$ such that the maximum mass remaining over all states of $p_{i'}$ is equal to $y$. Observe that showing $\rmassadv_{p_{i'}}(y)\ge x$ suffices to prove the lemma.

We can naively bound that the amount of mass remaining in the $j$-th state of $p_{i'}$ is at least $\min(y,p_{i'}(j)$, giving the bound $\sum_{j=1}^n\min(y, p_{i'}(j))\ge x$. However, this alone is not sufficient to obtain the above inequality. We can improve this bound by noting that the masses of the states allocated by the greedy algorithm are non-increasing. Hence for all $j$ with $p_{i'}(j)>y$, we must have previously subtracted at least $y$ from it. Therefore, the $j$-th state of $p_{i'}$ will have at least $p_{i'}(j)$ mass remaining if $p_{i'}(j) \le y$, and otherwise will have at least $\min(y,p_{i'}(j)-y)$ mass remaining. In total, this gives the bound
\begin{equation*}
\sum_{j=1}^n\begin{cases}
p_{i'}(j) & p_{i'}(j)\le y \\
\min(y, p_{i'}(j)-y) & p_{i'}(j)>y
\end{cases}\ge x.
\end{equation*}
We finish by observing that $\rmassadv_{p_{i'}}(x)$ is at least the LHS of this inequality.
\end{proof}

\begin{remark}
While $C_S$ corresponds to $\sketch_{\greed}$, $\rmassadv_{S}$ corresponds to $\isketch_{\greed}$.
\end{remark}

Finally, we prove \Cref{thm:1.22} by upper bounding the integral from \cref{lemma:integral} by $H(\iprof_S)$ plus a small constant.

\begin{lemma}
Let $C_S$ be as defined in \Cref{lemma:define-C}. Then
\begin{equation}
\int_0^1\lg \frac{1}{C_S(x)}\, dx\le H(\iprof_S)+\frac{1+\lg e}{2}.
\end{equation}
\end{lemma}

\begin{proof}
As $C_S$ is defined in terms of $\rmassadv_S$, we need to upper bound $\rmassadv_S(y)$ in terms of $\iprof_S$.  Defining $\isketch_p(y)\triangleq 1$ for $y>1$, we can check that
\begin{align}
\rmassadv_{p}(y)&=\isketch_{p}(y)+\frac{1}{2}(\isketch_{p}(2y)-\isketch_{p}(y))\nonumber\\
&\phantom{=}+y\int_{2y}^1\isketch'_{p}(t)/t\, dt  \nonumber\\
&=\frac{\isketch_{p}(y)+\isketch_{p}(2y)}{2}+ y\int_{2y}^1\isketch'_{p}(t)/t\, dt \label{eq:ds-rhs}
\end{align}

As \Cref{eq:ds-rhs} holds for all $p \in S$, it follows that $\rmassadv_S(x)$ is bounded above by the analogous expression with $\iprof_S$ in place of $\isketch_p$. Now we finish by computing the integral:
\begin{align}
\int_{0}^1 \lg \left(\frac{1}{C_{S}(x)}\right)\, dx&=\int_{0}^1{\rmassadv_S}'(y)\lg(1/y)\, dy \label{eq:multicouple-step1}\\
&=\rmassadv_S(y)\lg(1/y)\Big |_{y=0}^1+\int_0^1 \frac{\rmassadv_S(y)}{y\ln 2}\, dy\nonumber\\
&=\int_0^1 \frac{\rmassadv_S(y)}{y\ln 2}\, dy\nonumber\\
&\le \frac{1}{\ln 2}\int_0^1 \frac{\iprof_S(y)+\iprof_S(2y)}{2y}+\paren{\int_{2y}^1\iprof_S'(t)/t\, dt}\, dy\nonumber\\
&=\frac{1}{2}H(\iprof_S)+\frac{1}{2}H(\iprof_S)+\frac{1}{\ln 2}\int_{1/2}^1 \frac{1}{2y}\, dy+\frac{1}{2\ln 2}\int_0^1 \iprof'(y)\, dy\nonumber\\
&= H(\iprof_S)+\frac{1}{2}+\frac{1}{2\ln 2} \nonumber\\
&= H(\iprof_S)+\frac{1+\lg e}{2}.\nonumber
\end{align}

\end{proof}

\section{Coupling Guarantees for Small $m$}\label{section:small-m}

We need to generalize \cref{eq:coupling-2} to the setting of general $m$. Without loss of generality assume $p^t_1(1)\le p^t_2(1)\le \dots\le p^t_m(1)$. Then 
\begin{equation*}
\isketch_{p^{t+1}_i}(y)=\begin{cases}
    \isketch_{p^{t}_i}(y) & y < p^t_i(1)-p^t_1(1) \\
    \isketch_{p^{t}_i}(y)+p^t_i(1)-p^t_1(1) & p^t_i(1)-p^t_1(1) \le y < p^t_i(1) \\
    \isketch_{p^{t}_i}(y)-p^t_1(1) & p^t_i(1)\le y \\
    \end{cases}
\end{equation*}

and \cref{eq:coupling-2} becomes

\begin{equation*}
\iprof_{S^{t+1}}\:\begin{cases}
     \le \iprof_{S^t}(y)+p^t_{i}(1)-p^t_1(1) & p^t_{i}(1)-p^t_1(1)\le y < \min(p^t_{i+1}(1)-p^t_1(1), p^t_i(1)) \\
     =\iprof_{S^t}(y)-p^t(1)=\mass(S^{t+1}) & p^t_i(1)\le y \\
    \end{cases}
\end{equation*}

Letting $d_i\triangleq \frac{p_i^t(1)-p_1^t(1)}{p_1^t(1)}$ for $i\in [2,m]$ and $d_{m+1}=1$, we then find that the generalization of \cref{align:almost-done} is as follows:
\begin{align*}
M^{t+1}&\le M^t+\greed(t+1)\cdot \max_{0<d_2<d_3<\dots<d_{m+1}=1}\frac{1}{\ln 2}\sum_{i=2}^{m}d_i\ln(d_{i+1}/d_i)\\
&=M^t+\greed(t+1)\cdot \max_{0<d_2<d_3<\dots<d_{m+1}=1}\frac{1}{\ln 2}\sum_{i=2}^{m}d_i\paren{\ln(1/d_i)-\ln(1/d_{i+1})}.
\end{align*}
The quantity 
\begin{equation}
\sum_{i=2}^{m}d_i\paren{\ln(1/d_i)-\ln(1/d_{i+1})}\label{eq:general-m-expr}
\end{equation}
can be more intuitively interpreted as follows: ``Given $m-1$ rectangles indexed from $2\dots m$ all with lower-left corner at the origin, and the $i$-th has width $d_i$ and height $\ln(1/d_i)$, what is the maximum possible area of their union?''

\begin{remark}
As $m\to\infty$, \cref{eq:general-m-expr} approaches $\frac{1}{\ln 2}\int_0^1\ln (1/x)\, dx=\frac{1}{\ln 2}\paren{x-x\ln x}\Big|_0^1=\log e \approx 1.44$.
\end{remark}

\begin{table}[H] 
\renewcommand{\arraystretch}{2}
\centering
\caption{Best-Known Additive Approximation Guarantees} \label{table:precise-guarantees}
\footnotesize
\begin{tabular}{|c|c|c|}
\hline
\textbf{} & \textbf{Prior Work}& \textbf{This Work} \\
\hline
{$m=2$}& $1$ \cite{cicalese2019minimum,rossi2019greedy,li2021efficient} & ${\frac{1}{e}\log_2(e) \approx 0.53}$ \\
\hline
{$m=3$}& $\log_2(e) \approx 1.44$ \cite{compton2022tighter}  & ${\max_{0<d_{2}<d_{3}<d_{4}=1}\sum_{i=2}^{3}d_i\log_2(d_{i+1}/d_i) \approx 0.77}$ \\
\hline
{$m=4$}& $\log_2(e) \approx 1.44$ \cite{compton2022tighter}  & ${\max_{0<d_{2}<\dots<d_{5}=1}\sum_{i=2}^{4}d_i\log_2(d_{i+1}/d_i) \approx 0.90}$ \\
\hline
{$m=5$}& $\log_2(e) \approx 1.44$ \cite{compton2022tighter}  & ${\max_{0<d_{2}<\dots<d_{6}=1}\sum_{i=2}^{5}d_i\log_2(d_{i+1}/d_i) \approx 0.99}$ \\
\hline
{$m=6$}& $\log_2(e) \approx 1.44$ \cite{compton2022tighter}  & ${\max_{0<d_{2}<\dots<d_{7}=1}\sum_{i=2}^{6}d_i\log_2(d_{i+1}/d_i) \approx 1.06}$ \\
\hline
{$m=7$}& $\log_2(e) \approx 1.44$ \cite{compton2022tighter}  & ${\max_{0<d_{2}<\dots<d_{8}=1}\sum_{i=2}^{7}d_i\log_2(d_{i+1}/d_i) \approx 1.10}$ \\
\hline
{$m=8$}& $\log_2(e) \approx 1.44$ \cite{compton2022tighter}  & ${\max_{0<d_{2}<\dots<d_{9}=1}\sum_{i=2}^{8}d_i\log_2(d_{i+1}/d_i) \approx 1.14}$ \\
\hline
{$m=9$}& $\log_2(e) \approx 1.44$ \cite{compton2022tighter}  & ${\max_{0<d_{2}<\dots<d_{10}=1}\sum_{i=2}^{9}d_i\log_2(d_{i+1}/d_i) \approx 1.17}$ \\
\hline
{$m=10$}& $\log_2(e) \approx 1.44$ \cite{compton2022tighter}  & ${\max_{0<d_{2}<\dots<d_{11}=1}\sum_{i=2}^{10}d_i\log_2(d_{i+1}/d_i) \approx 1.19}$ \\
\hline
{$m=11$}& $\log_2(e) \approx 1.44$ \cite{compton2022tighter}  & ${\max_{0<d_{2}<\dots<d_{12}=1}\sum_{i=2}^{11}d_i\log_2(d_{i+1}/d_i) \approx 1.21}$ \\
\hline
{General $m$}& $\log_2(e) \approx 1.44$ \cite{compton2022tighter} & ${\frac{1}{2}(\log_2(e) + 1)\approx 1.22}$ \\
\hline
\end{tabular}
\label{tab4}
\end{table}

\section{Experiment Details}\label{section:experiment-details}
Experiments were run on a laptop with a 3.49GHz 8-Core Apple M2 processor. Each entry of \cref{table:experiment} corresponds to the average of 100 runs over the specified configuration. If any of the runs exceeded 120 seconds, the entry is marked as ``$>120$''. In \cref{tab5}, we provide a corresponding experiment when coupling two distributions with different numbers of states $n_1,n_2$. In \cref{tab3}, we provide the same data as \cref{table:experiment}, but accompanied with the standard deviation. Each input distribution was sampled from the Dirichlet distribution with parameter 1 (meaning, uniform over the simplex). The algorithm ``Naive Polytope Vertex Enumeration'' corresponds to naively enumerating over all vertices in the corresponding optimization polytope. We were particularly generous to this baseline, and implemented it using insights regarding induced spanning trees (that were not discussed before our work) to avoid requiring it to solve systems of linear equations. Essentially, this would run even slower if done extremely naively. Rows of the form ``Backtracking [X]'' denotes using the backtracking algorithm of \cref{app:back} with the lower bound $X$ for its subroutine that helps prune the search space. The algorithm ``Dynamic Programming'' corresponds to \cref{app:dp}. All code is implemented in C++.

\begin{table*}[htbp] 
\centering
\caption{Average Runtime and Standard Deviation (in Seconds) for Exactly Coupling Two Distributions} \label{table:experiment-with-stddev}
\tiny
\setlength\tabcolsep{1.5pt}
\begin{tabular}{|c|c|c|c|c|c|c|c|c|}
\hline
\textbf{}&\multicolumn{8}{|c|}{\pmb{$n$}} \\
\cline{2-9}
\textbf{Algorithm}& 4 & 5 & 6 & 7 & 8 & 9 & 10 & 11\\
\hline
{Naive Polytope Vertex Enumeration} & $0.002\pm 0.0006$& $0.189\pm 0.002$& $33.79\pm 0.240$& $>120$& $>120$& $>120$& $>120$& $>120$\\
\hline
{Backtracking [$0$]} & $0.0003\pm 0.0001$& $0.004\pm 0.002$& $0.106\pm 0.041$& $3.576\pm 1.839$& $>120$& $>120$& $>120$& $>120$\\
\hline
{Backtracking [$H(\andd S)]$} & $0.0001\pm 0.00006$& $0.001\pm 0.0005$& $0.013\pm 0.010$& $0.256\pm 0.196$& $4.907\pm 3.896$& $>120$& $>120$& $>120$\\
\hline
{Backtracking [$H(\profile_S)$]} & $0.0001\pm 0.00006$ & $0.0009\pm 0.0005$ & $0.009\pm 0.008$& $0.153\pm 0.133$& $2.206\pm 1.995$& $>120$& $>120$& $>120$\\
\hline
{Backtracking [$H(\majorpro_S)$]} & $0.00006\pm 0.00003$& $0.0004\pm 0.0002$& $0.003\pm 0.003$& $0.035\pm 0.035$& $0.344\pm 0.352$& $5.398\pm 7.761$& $>120$& $>120$\\
\hline
{Dynamic Programming} & $0.00007\pm 0.00002$& $0.0004\pm 0.00002$& $0.002\pm 0.00004$& $0.012\pm 0.0003$& $0.093\pm 0.001$& $0.802\pm 0.016$& $9.769\pm 0.278$& $>120$\\
\hline
\end{tabular}
\label{tab3}
\end{table*}

\begin{table*}[htbp] 
\centering
\caption{Average Runtime (in Seconds) for Exactly Coupling Two Distributions with Different Cardinalities} 
\scriptsize
\setlength\tabcolsep{1.5pt}
\begin{tabular}{|c|c|c|c|c|c|c|c|c|c|c|c|}
\hline
\textbf{}&\multicolumn{11}{|c|}{\pmb{$n_1,n_2$}} \\
\cline{2-12}
\textbf{Algorithm}& 6,3 & 7,3 & 8,3 & 9,3 & 10,4 & 11,4 & 12,4 & 13,5 & 14,5 & 15,5 & 16,6 \\
\hline
{Naive Polytope Vertex Enumeration} & $0.0003$ & $0.038$& $0.219$& $1.172$& $>120$ & $>120$& $>120$& $>120$& $>120$& $>120$ & $>120$\\
\hline
{Backtracking [$0$]} & $0.0001$ & $0.002$& $0.008$& $0.030$& $1.147$ & $4.538$& $19.841$& $>120$& $>120$& $>120$ & $>120$\\
\hline
{Backtracking [$H(\andd S)]$} & $0.00008$& $0.0008$& $0.003$& $0.010$& $0.157$ & $0.498$& $2.139$& $>120$& $>120$& $>120$ & $>120$\\
\hline
{Backtracking [$H(\profile_S)$]} & $0.00009$& $0.0008$& $0.003$& $0.008$& $0.104$ & $0.295$& $1.066$& $>120$& $>120$& $>120$ & $>120$\\
\hline
{Backtracking [$H(\majorpro_S)$]} & $0.00006$& $0.0005$& $0.001$& $0.005$& $0.046$ & $0.139$& $0.566$& $5.855$& $19.927$& $>120$ & $>120$\\
\hline
{Dynamic Programming} & $0.00005$& $0.0003$ & $0.0008$& $0.002$& $0.011$ & $0.031$& $0.084$& $0.724$& $2.509$& $8.984$ & $>120$ \\
\hline
\end{tabular}
\label{tab5}
\end{table*}

\section{Algorithms for Section \ref{section:exactly-optimal}: Exactly Coupling Two Distributions, and Efficient Lower Bounds}

\subsection{Dynamic Programming}\label{app:dp}

First, we note that by \cref{lemma:monotone-backtracking}, any minimum-entropy coupling of $\p$ can be associated with some spanning forest of the unweighted graph with vertex set $V=\{\ell_{1\dots n}, r_{1\dots n}\}$; then we can add 0-weight edges to change this forest into a spanning tree. 

Our dynamic programming solution searches over all possible spanning trees. It keeps track of $O(2^{2n}\cdot n)$ states: for every $S\subseteq V$ and $v \in S$, we store the minimum entropy over all partial couplings of $S$ such that $v$ is the only vertex in $S$ with any mass remaining in $dp[S][v]$. Note that the remaining mass of $v$ must equal $\sum_{w\in S}(-1)^{\text{side}(w)\neq \text{side}(v)}\text{orig\_mass}(w)$ across all such partial couplings. Effectively, the state $(S,v)$ corresponds to all partial couplings that induce a spanning tree on $S$ rooted at $v$.

There are two types of transitions between states: the first attaches a new root to a spanning tree (taking $O(2^{2n}\cdot n^2)$ time), while the second merges two spanning trees with the same root (taking $O(3^{2n}\cdot n)$ time). The minimum entropy over all couplings is $dp[V][v]$ for any $v\in V$.

\subsection{Backtracking}\label{app:back}

\begin{algorithm}[H]
\begin{small}
    \caption{Backtracking Algorithm}
  \label{alg:backtrack}
\begin{algorithmic}[1]
    \State {\bfseries Input:} Marginal distributions of $m=2$ variables each with $n$ states $\{\mat{p},\mat{q}\}$.
    \State Initialize $best\_entropy = \infty$.
    \State \textbf{Backtrack}($\{\mat{p},\mat{q}\},entropy\_so\_far = 0, last\_mass = \infty$).
    \State \Return $best\_entropy$.
	\Procedure{\bfseries Backtrack($\{\mat{p},\mat{q}\}, entropy\_so\_far, last\_mass$)}{}
	\If {$entropy\_so\_far + X(\{\mat{p},\mat{q}\}) \ge best\_entropy$} \label{line:lb}
	\State \Return
	\EndIf
	\If {Mass$(\{\mat{p},\mat{q}\})=0$}
	\State $best\_entropy = \min(best\_entropy, entropy\_so\_far)$.
	\State \Return
	\EndIf
	
	\For {$(i,j)$ s.t. $0 < \min(\mat{p}(i),\mat{q}(j)) \le last\_mass$}
	\State $m = \min(\mat{p}(i),\mat{q}(j))$. $\mat{p'} = \mat{p}$. $\mat{q'} = \mat{q}$.
	\State $\mat{p'}(i) -\!= m$. $\mat{q'}(j) -\!= m$.
	\State \textbf{Backtrack}($\{\mat{p'},\mat{q'}\},entropy\_so\_far + m \log(1/m), m$).
	\EndFor
	\State \Return
	\EndProcedure
\end{algorithmic}
\end{small}
\end{algorithm}

Now we argue that this algorithm is correct. First, note that line 6 does not affect the correctness of the algorithm as long as $X(\{p,q\})\le \textsc{OPT}(\{p,q\})$, so we can ignore it. Next, we show that as long as $\p(i,j)$ is some minimum entropy coupling, our algorithm will construct it. Let $(i,j)$ be some pair such that $\p(i,j)$ is maximized. Then by \cref{lemma:monotone-backtracking}, our algorithm will compute $m=\min(p(i),q(j))$ on the $(i,j)$-th iteration of step 4. Then by induction, our algorithm will construct the remainder of $\p$ after recursing into $\{p',q'\}$ on that loop iteration.

\subsection{Lower Bounds in Almost-Linear Time}\label{app:lower-algs}

\paragraph{Computing $\profile_S$.}

Recall from \Cref{def:sketch} that for all $p\in S$ and $i\le |p|$, we must have $\profile_S\paren{\sum_{j\ge i}p(j)}\le p(i)$. We can compute $\profile_S$ using the following procedure:

\begin{enumerate}
    \item Define $T=\{(0,0)\} \cup \{(\sum_{j\ge i}p(j), p(i)) \mid p\in S, 1\le i\le |p|\}$.
    
    \item For each pair $(x,y)\in T$ in decreasing order of $x$, if $y$ is greater or equal to the smallest $y$ seen so far, remove the pair from $T$. In other words, there exists another pair $(x',y')\in T$ such that $x'\ge x$ and $y'\le y$, so the pair $(x,y)$ is redundant and removing it does not affect the profile.
    
    \item Reverse $T$. Now the entries of $T$ are increasing in both $x$ and $y$.
\end{enumerate}

Now $H(\profile_S)$ is the sum of $(x_2-x_1)\lg(1/y_2)$ over all adjacent pairs $(x_1,y_1)$ and $(x_2,y_2)$ in $T$. The time complexity is $\tilde O(nm)$.

\paragraph{Computing $\majorpro_S$.} Assume we have already computed $T=[(x_1,y_1),\dots,(x_{|T|},y_{|T|})]$ as in the previous part. Then it suffices to describe how to implement the algorithm described in \Cref{section:major-pro-proof} in $O(|T|)$ additional time. We recall that the algorithm essentially asks us to answer the following query for at most $|T|$ values of $t$ in decreasing order: how large can the side length of a square with lower-right corner at $(t,0)$ be without the square exceeding the profile?

We first describe how to answer each query in $O(|T|)$ time, resulting in an $O(|T|^2)$ algorithm. If we are growing a square with lower-right corner at $(t,0)$, then its side length must be
\begin{equation}
\min_{(x,y)\in T}\max(y, t-x)=\min_{(x,y)\in T}\begin{cases}
t-x & x+y\le t\\
y & x+y > t
\end{cases},  \label{eq:maj-profile}  
\end{equation}
which can be evaluated in $O(|T|)$ time.

We can speed up the above algorithm by noting that $T$ is sorted in increasing order of $x+y$. Then define $j$ be the minimum index such that $x_j+y_j>t$. Now \Cref{eq:maj-profile} simplifies to: 
\begin{equation*}
\min_{(x,y)\in T}\begin{cases}
t-x & x+y\le t\\
y & x+y > t
\end{cases}=\min\paren{\min_{1\le i<j}t-x_i,\min_{j\le i\le |T|}y_j}=\min(t-x_{j-1},y_j).
\end{equation*}
As the queries come in decreasing order of $t$, $j$ must stay the same or decrease between queries. Thus the time required to answer all queries is proportional to the number of queries plus the number of times $j$ moves between queries, both of which are $O(|T|)$.

\section{Proofs for Section \ref{section:exactly-optimal}: Exactly Optimal Solutions}\label{section:exactly-optimal-proof}

\subsection{Proof of \cref{theorem:np-reduct}}

\npreduct*

By \cref{lemma:small-sup} it suffices to provide the coordinates of some vertex $v\in \BP$ at which the objective is minimized. We note that \cref{lemma:small-sup} shows that every vertex has at most $e\le nm-m+1$ nonzero coordinates. Furthermore, if the input probabilities are rational numbers with bounded bit complexity, then $v'$ (the vector consisting only of the nonzero entries of $v$) must be the unique solution to a linear system of the form $Av'=b$, where $A$ consists solely of zeros and ones, and $A$ has shape $e\times e$. The solution to this system can easily be seen to have bit complexity bounded above by the bit complexity of $\p$ times $\text{poly}(nm)$.

\subsection{Proof of \cref{lemma:monotone-backtracking}}

\monotoneback*

For the first part of the lemma, it suffices to check that $G$ cannot contain any cycles. Suppose for the sake of contradiction that $G$ contains a simple cycle $(i_1,j_1,i_2,j_2,\dots,i_k,j_k)$, and define $\delta>0$ to be the minimum weight along the cycle. Then because entropy is concave, we can obtain a smaller entropy by alternately adding and subtracting $\delta$ along the edges of the cycle.

For the second part, suppose for the sake of contradiction that there exists an edge $(i,j)$ of the maximum weight such that $i$ is connected to some other vertex $j'$ and $j$ is connected to some other vertex $i'$ in $G$. Let the weights $w(i,j)$, $w(i,j')$, and $w(i',j)$ be $a$, $b$, and $c$, respectively; WLOG $a\ge b\ge c$. Then we claim that the entropy would decrease if we perform the following four modifications: 
\begin{enumerate}
    \item increase $w(i,j)$ by $c$
    \item decrease $w(i,j')$ by $c$
    \item increase $w(i',j')$ by $c$
    \item set $w(i',j)=0$
\end{enumerate}

Indeed, the change in entropy is upper bounded by
\begin{align*}
&\left[(a+c)\lg(1/(a+c))+(b-c)\lg(1/(b-c))+c\lg(1/c)\right]-\left[a\lg(1/a)+b\lg(1/b)+c\lg(1/c)\right]\\
&=\left[(a+c)\lg(1/(a+c))+(b-c)\lg(1/(b-c))\right]-\left[a\lg(1/a)+b\lg(1/b)\right],
\end{align*}

which is negative because because $(a+c,b-c)$ majorizes $(a,b)$.

\subsection{Proof of \cref{theorem:major-pro}}
\label{section:major-pro-proof}

\majorprothm*

We first describe how to construct $\majorpro_S$. In general, we define $\majorpro_S(i)$ in terms of $\majorpro_S(1)$, $\majorpro_S(2)$, $\dots$, $\majorpro_S(i-1)$. Specifically, if $\sum_{j=1}^{i-1}\majorpro_S(j)<\mass(S)$, then $\majorpro_S(i)$ is defined to be the maximum positive real $r_i$ such that $\profile_S(x)\ge r_i$ for all $x>\mass(S)-\sum_{j=1}^{i-1}\majorpro_S(j)-r_i$. Intuitively, this corresponds to growing a square with side length $r_i\times r_i$ and lower-right corner at $(1-\sum_{j=1}^{i-1}\majorpro_S(j),0)$ until it touches $\profile_S$. Note that $r_{i+1}>r_i$ would lead to a contradiction, because setting $r_i=r_{i+1}$ would lead to the square associated with $r_i$ lying below the profile, meaning that $r_i$ could have been larger. Thus $r_i\ge r_{i+1}$, meaning that our construction produces the elements of a probability distribution in non-increasing order.

Next, we show that $\majorpro_S$ majorizes any distribution $p$ such that $\sketch_{p}(x)\le \profile_S(x)$ for all $x$ using induction. Suppose that we have already proven that $\sum_{j=1}^ip(j)\le \sum_{j=1}^i\majorpro_S(j)$; it remains to show this inequality for $i+1$. Indeed, the inequality for $i+1$ rearranges to
\begin{equation*}
r_i\ge \sum_{j=1}^{i+1}p(j)-\sum_{j=1}^{i}\majorpro_S(j).
\end{equation*}
This inequality is true because the entirety of the square with lower-right corner at $(1-\sum_{j=1}^{i}\majorpro_S(j),0)$ and side length $\max\paren{\sum_{j=1}^{i+1}p(j)-\sum_{j=1}^{i}\majorpro_S(j),0}$ is contained within the square with lower-right corner at $(\sum_{j=1}^{i}p(j),0)$ and side length $p(i+1)$. The entirety of this second square lies below $\profile_S$ by the assumption on $p$. 

Now we are ready to show the inequalities in \cref{theorem:major-pro}.
\begin{itemize}
    \item Because $\sketch_{\opt}$ lies below $\profile_S$, the work above shows that $\opt$ is majorized by $\majorpro_S$. Thus, $H(\opt)\ge H(\majorpro_S)$.
    \item Next, 
    \begin{align*}
    H(\majorpro_S)&=\int_{0}^1\lg\paren{1/\sketch_{\majorpro_S}(x)}\, dx\\
    &\ge \int_{0}^1\lg\paren{1/\profile_S(x)}\, dx\\
    &=H(\profile_S),
    \end{align*}
    where the inequality holds because $\sketch_{\majorpro_S}(x)\le \profile_S(x)$ for all $x$.
    \item Finally, we claim that $H(\majorpro_S)\ge H(\andd S)$ because $\andd S$ majorizes $\majorpro_S$. Suppose for the sake of contradiction that $\andd S$ does not majorize $\majorpro_S$; then there exists $p\in S$ such that $p$ does not majorize $\majorpro_S$. This in turn implies that there exists some $i$ such that $\sum_{j=1}^ip(j)\ge \sum_{j=1}^i\majorpro_S(j)$ and $\sum_{j=1}^{i+1}p(j)< \sum_{j=1}^{i+1}\majorpro_S(j)$. But then there exists some $x$ such that $\sketch_{\majorpro_S}(x)>\sketch_p(x)\ge \profile_S(x)$, contradicting the assumption that $\sketch_{\majorpro_S}(x)\le \profile_S(x)$ for all $x$. Thus, $\andd S$ must majorize $\majorpro_S$.
\end{itemize}

\begin{remark}

It is natural to ask whether 
we can in fact show that $H(\profile_S)\ge H(\andd_S)$, particularly given that this inequality often holds for randomly generated instances (e.g. all the runs included in \cref{tab3}). It turns out that this is not true in general; $S=\{[0.5, 0.5], [0.75, 0.05, 0.05, 0.05, 0.05, 0.05]\}$ is a counterexample.

\end{remark}

\section{Optimality Gaps: Bounds on Possible Improvement}\
\label{section:gaps}

In this section, we examine the open landscape of upper bounds for greedy coupling in reference to particular lower bounds. In \cref{tab:greed-ub}, we recall the best-known proven upper bounds on the entropy of the greedy coupling with respect to particular lower bounds. Whether or not these bounds can be improved (and if so, by how much) are important open questions. In the first column (``Lower Bound on $H(\greed)$'') of \cref{tab:greed-opt-lb}, we detail counter-examples that show limits for how much the upper bounds of the greedy algorithm in \cref{tab:greed-ub} can be improved. In the second column (``Lower Bound on $H(\opt)$'') of \cref{tab:greed-opt-lb}, we detail counter-examples that show limits for how much the upper bounds of \emph{any algorithm} can be improved. The counter-examples are a mixture of results from prior work, new hand-constructed instances, and instances we computationally discovered with a local search. Note that the entries in the first and fifth row of the second column of \cref{tab:greed-opt-lb} correspond to the majorization barrier. Finally, in \cref{sec:other-algs} we discuss counter-examples that show the algorithms of \cite{cicalese2019minimum,li2021efficient} for coupling $m=2$ distributions do not match the approximation guarantee for the greedy coupling we show in \cref{thm:0.53}.

We emphasize that for all upper bounds in \cref{tab:greed-opt-lb} it is open whether they can be improved (other than $H(\andd_S)$ for general $m$, which is tight). We note that this section shows how the gap for improving the upper bound of $H(\greed)-H(\opt)$ for $m=2$ is comparatively small, as we are able to provide an input where $H(\greed)-H(\opt)=0.40$, which is relatively close to the upper bound provided by \cref{thm:0.53} ($\approx 0.53$). On the other hand, for general $m$, we were unable to generate instances with $H(\greed)-H(\opt)>0.71$, which is relatively far away from the upper bound provided by \cref{thm:1.22} ($\approx 1.22$). We did not bound $H(\opt)-H(\profile_S)$ or $H(\opt)-H(\majorpro_S)$ for general $m$ due to how it is computationally infeasible to compute optimal solutions for general $m$ (making local search similarly infeasible).

\emph{Please note that all lower bounds in \cref{tab:greed-opt-lb} are simply the best counter-examples that are currently known (some by a fairly limited computational search), and do not correspond to any conjecture of the best lower bound.}

\begin{table}[H]
    \centering
    \begin{tabular}{cc|c}
       Reference Lower Bound & $m$ & Best Upper Bound on $H(\greed)$ \\
    \hline
      $H(\andd_S)$      & $m=2$  & $H(\greed)\le H(\andd_S)+1$ \cite{rossi2019greedy}  \\
      $H(\profile_S)$   & $m=2$ & $H(\greed)\le H(\profile_S)+0.53$ (\cref{thm:0.53}) \\
      $H(\majorpro_S)$  & $m=2$ & same as above \\
      $H(\opt)$ & $m=2$ & same as above \\
      $H(\andd_S)$      & general $m$ & $H(\greed)\le H(\andd_S)+1.44$ \cite{compton2022tighter} \\
      $H(\profile_S)$   & general $m$ & $H(\greed)\le H(\profile_S)+1.22$ (\cref{thm:1.22}) \\
      $H(\majorpro_S)$  & general $m$ & same as above \\
      $H(\opt)$ & general $m$ &  same as above \\
    \end{tabular}
    \caption{Upper Bounds on Greedy}
    \label{tab:greed-ub}
\end{table}

\begin{table}[H]
    \centering
    
\scriptsize
    \begin{tabular}{cc|cc}
       Reference Lower Bound & $m$ & Lower Bound on $H(\greed)$ & Lower Bound on $H(\opt)$ \\
    \hline
      $H(\andd_S)$      & $m=2$  & $H(\greed)\approx H(\andd_S)+0.66$  (\cref{sec:local-search})  & $H(\opt)\approx H(\andd_S)+0.66$  (\cref{sec:local-search}) \\
      $H(\profile_S)$   & $m=2$ & $H(\greed)\approx H(\profile_S)+0.46$ (\cref{sec:greed-profile}) & $H(\opt)\approx H(\profile_S)+0.39$ (\cref{sec:local-search}) \\
      $H(\majorpro_S)$  & $m=2$ & same as above & $H(\opt)\approx H(\majorpro_S)+0.35$ (\cref{sec:local-search}) \\
      $H(\opt)$ & $m=2$ &  $H(\greed)= H(\opt)+0.40$ (\cref{sec:local-search}) & $=$ \\
      $H(\andd_S)$      & general $m$ & $H(\greed)\approx H(\andd_S)+1.44$ \cite{compton2022tighter} & $H(\opt)\approx H(\andd_S)+1.44$ \cite{compton2022tighter} \\
      $H(\profile_S)$   & general $m$ & $H(\greed)\approx H(\profile_S)+0.89$ (\cref{sec:greed-profile}) & same as $m=2$ \\
      $H(\majorpro_S)$  & general $m$ & same as above & same as $m=2$  \\
      $H(\opt)$ & general $m$ &  $H(\greed)\approx H(\opt)+0.71$ (\Cref{sec:rand-gen}) & $=$ \\
    \end{tabular}
    \caption{Lower Bounds on Greedy and OPT}
    \label{tab:greed-opt-lb}
\end{table}

\subsection{Gaps Between $H(\greed)$ and $H(\profile_S)$}\label{sec:greed-profile}

Our lower bounds are achieved by letting $\profile_S=\sketch_{p^*}$, where $p^*\in S$ is a uniform distribution. We note that $H(\profile_S)=H(\majorpro_S)=H(p^*).$

Let $U_s$ denote the uniform distribution over $s$ states. 

\paragraph*{Case $m=2$.} We used $S=\{U_{F_t},U_{L_{t-1}}\}$. Here, $F_n=\text{round}\paren{\frac{\vphi^n}{\sqrt 5}}$ denotes the $n$-th Fibonacci number, while $L_n=\text{round}\paren{\vphi^n}$ denotes the $n$-th Lucas number. For $t=12$, we have $F_t=144$, $L_{t-1}=199$, and $H(\greed)-H(p^*)>0.457$.

\paragraph*{Case general $m$.} We started with $S=\{U_1,U_2,\dots, U_{n}\}$, where $n=2000$. This construction already gives $H(\greed)-H(p^*)>0.805$. To construct instances where $H(\greed)-H(p^*)$ is even greater, we first state an inequality characterizing the behavior of the greedy algorithm on $S$:
\begin{equation}
\greed(t+1) \ge \min\paren{\frac{1-\sum_{i=1}^t\greed(i)}{\min(t, n)}, \frac{1}{n}}.\label{eq:better-guarantee}
\end{equation}
The first expression in the right hand side of \Cref{eq:better-guarantee} comes from all $p\in S$ with $|p|\le t$, while the second expression comes from all $p\in S$ with $|p|>t$. Thus, $\greed$ majorizes the distribution $\greed'$ satisfying $\greed'(t+1)=\min\paren{\frac{1-\sum_{i=1}^t\greed'(i)}{\min(t, n)}, \frac{1}{n}}$ for all $t\ge 0$. If we add to $S$ a distribution $p_t$ with $|p_t|=t$ to make \Cref{eq:better-guarantee} tight for each $t\in (1000, 1414]$ , then we obtain $H(\greed)-H(p^*)>0.887$. It is likely possible to construct instances with even greater gaps between $\greed$ and $p^*$, as
\begin{equation}
H(\greed)-H(\profile_S)\le H(\greed')-H(\profile_S)\approx 1.082. \label{eq:greed-guarantee}
\end{equation}
However, \Cref{eq:better-guarantee} cannot be made tight for $t=1415$, so we cannot have equality in \Cref{eq:greed-guarantee}. It is unclear what the value of $H(\greed)-H(\profile)$ is if we let $S$ consist of \textit{all} distributions with sketches lying above $\sketch_{p^*}$. %

\subsection{Local Search for $m=2$}\label{sec:local-search}

In this section, we detail ``counter-examples'' found with a computational local search, i.e. example distributions achieving large gaps between various coupling entropy bounds. These example empirical gaps provide limits on any future theoretical bounds on the size of these gaps. At a high-level, this search works by repeatedly making random perturbations to the input distributions and accepting changes that increase the gap we are finding a lower bound for. %

\underline{$H(\greed)-H(\andd_S) \approx 0.662463$}\\
$p_1 = [0.3199940773, 0.3199844734, 0.1200540976, 0.1200022716, 0.1199650801]$\\
$p_2 = [0.2000218248, 0.2000211548, 0.2000202369, 0.1999737730, 0.1999630105]$

\underline{$H(\opt)-H(\andd_S) \approx 0.662405$.}\\
$p_1 = [0.3199940773, 0.3199844734, 0.1200540976, 0.1200022716, 0.1199650801]$\\
$p_2 = [0.2000218248, 0.2000211548, 0.2000202369, 0.1999737730, 0.1999630105]$

\underline{$H(\opt)-H(\profile_S) \approx 0.389941$.}\\
$p_1 = [0.2128275903, 0.2122898591, 0.2119627146, 0.2119384365, 0.1509813995]$\\
$p_2 = [0.2747693214,0.2739951951, 0.2739769942, 0.1161585898, 0.0610998995]$

\underline{$H(\opt)-H(\majorpro_S) \approx 0.354485$.}\\
$p_1 = [0.2128275903, 0.2122898591, 0.2119627146, 0.2119384365, 0.1509813995]$\\
$p_2 = [0.2747693214,0.2739951951, 0.2739769942, 0.1161585898, 0.0610998995]$

\underline{$H(\greed) 
 - H(\opt)\approx 0.395053$.}\\
$p_1 = [0.4081266587, 0.3060949942, 0.1530474970, 0.0765237476, 0.0382618746, 0.0179452279]$ \\
$p_2 = [0.3060949942, 0.2040633294, 0.2040633294, 0.1530474970, 0.0765237476, 0.0382618746, 0.0179452278]$

We can modify this last example to obtain $H(\greed)-H(\opt)=0.4$ exactly. Specifically, we can let\\
$p_1=[0.4,0.3,0.15,0.075,0.0375,\dots]$\\
$p_2=[0.3,0.2,0.2,0.15,0.075,0.0375,\dots]$

where the dots represent a geometric sequence with ratio $0.5$. Then $\opt=p_2$, but $\greed$ replaces one occurrence of $0.2$ in $\opt$ with $[0.1, 0.05, 0.025, \dots]$.

\subsection{Gap Between $H(\greed)$ and $H(\opt)$ for General $m$}\label{sec:rand-gen}

We can generate instances where $H(\greed)$ and $H(\opt)$ differ by starting with $S=\{\opt\}$ and repeatedly inserting into $S$ any probability distribution generated by combining states of $\opt$. We generated $2000$ instances with $|\opt|=10$ from the simplex and $m=4001$. The instance that resulted in the maximum difference between $H(\greed)$ and $H(\opt)$ gave $H(\greed)\approx H(\opt)+0.707147$.

\subsection{Counter-Examples for Other Algorithms}\label{sec:other-algs}
The works of \cite{cicalese2016approximating,li2021efficient} introduce algorithms for coupling $m=2$ distributions. Prior to our work, the algorithms of \cite{kocaoglu2017entropic,cicalese2019minimum,li2021efficient} all had the best additive guarantee of $1$ bit. By \cref{thm:0.53}, we show the greedy coupling algorithm of \cite{kocaoglu2017entropic} is always within $\log(e)/e \approx 0.53$ bits of the optimum. We provide a counter-example where the algorithms of \cite{cicalese2016approximating,li2021efficient} do not meet this approximation guarantee (found with local search). We do not conjecture whether these algorithms have counter-examples where their optimality gap is larger. Interestingly, for this counter-example, the algorithms of \cite{cicalese2016approximating,li2021efficient} both produce a coupling with the same output, where both algorithms have an optimality gap of $\approx 0.6265549682$ for: \\
$p_1 = [0.5540050843, 0.1984459548, 0.1288780486, 0.0396890356, 0.0789189118, 0.0000629649]$ \\
$p_2 = [0.2770967899, 0.2769100227, 0.1984729975, 0.1288783386, 0.0789194408, 0.0397224105]$ 

\section{Generalizing to Concave Minimum-Cost Couplings}\label{section:concave-costs}

We emphasize that our methods in all previous sections are not unique to the entropy function but broadly apply to concave functions. These methods also have the flexibility to yield multiplicative guarantees for functions where an additive approximation guarantee is unexpected. We will define the cost function for coupling with distribution $p$ as $\Fc(p) \triangleq \sum_i \fc(p(i))$. The necessary assumptions are that $\fc$ is non-negative and concave and $\fc(0)=0$. Note that entropy is a specific example of this where $\fc(x) = x \lg(1/x)$. The corresponding lower bound given by the profile is $\Fc(\profile_S)=\int_0^1 \frac{\fc(\profile_S(x))}{\profile_S(x)}\,dx$. %

For cost functions of the form $\fc(x)=x^c$ where $c\in (0,1)$, we can use the same methods as \cref{section:two-couple} and \cref{section:many-couple} to show that the cost achieved by the greedy algorithm is within a multiplicative factor of the optimum.

\begin{theorem}\label{thm:mult-approx-2}

For $m=2$, let $r\triangleq\max_{0<q<p\le 1}\frac{\fc(q)}{\fc(p)}-\frac{q}{p}$. If $r<1$, then the greedy algorithm achieves a $\frac{1}{1-r}$-multiplicative approximation.

\end{theorem}

\begin{theorem}\label{thm:mult-approx-general}

For general $m$, when $\fc(x)=x^c$, the greedy algorithm obtains a $\paren{\frac{1}{2}+\frac{1}{c2^c}}$-multiplicative approximation.

\end{theorem}

\begin{corollary}\label{corollary:sqrt-approx}

For $\fc(x)=\sqrt x$, the greedy algorithm obtains a $4/3$-multiplicative approximation for $m=2$ and a $\approx 1.91$-multiplicative approximation for general $m$.

\end{corollary}

\begin{proof}

For $m=2$, substituting $\fc$ into \cref{thm:mult-approx-2}, we obtain $r=(q/p)^{0.5}-q/p$. This is maximized when $q/p=1/4$, which gives us $r=1/4$. Thus, the greedy algorithm gives us a $\frac{1}{1-1/4}=\frac{4}{3}$-multiplicative approximation. For general $m$, we simply substitute $c=0.5$ into \cref{thm:mult-approx-general}, giving a $(\frac{1}{2}+\sqrt 2)$-multiplicative approximation.

\end{proof}

We note that the part of \Cref{corollary:sqrt-approx} corresponding to $m=2$ can easily be generalized to other cost functions of the form $\fc(x)=x^c$ for $c\in (0,1)$.   For general $m$, one could obtain guarantees for general concave costs by integrating over $\int_0^1 \rmassadv_S(y) \cdot \frac{\fc(y)}{y} \, dy$, analogous to \cref{equation:rmass-greed-bound}. %

\subsection{Proof of \cref{thm:mult-approx-2}}

We first extend the definition of $\Fc$ to inverse sketches (and profiles):
\begin{align}
\Fc(\isketch_p)&\triangleq \isketch_p(1)\fu(1)+\int_{0}^1\isketch_p(y)(-\fu'(y))\,dy \label{eq:fc-def-1} \\
&=\int_0^1 \isketch_p'(y)\fu(y)\, dy. \label{eq:fc-def-2}
\end{align}
We can verify that \cref{eq:fc-def-1} is consistent with $\Fc(p)$:
\begin{align*}
\Fc(\isketch_p)&=\sum_{j=1}^np_j\paren{\fu(1)+\int_{p_j}^1(-\fu'(y))\,dy}\\
&=\sum_{j=1}^np_j\fu(p_j)=\sum_{j=1}^n\fc(p_j)=\Fc(p).
\end{align*}
The proof is very similar to that of \Cref{thm:0.53} after replacing all occurrences of $H$, $x\ln(1/x)$, and $\ln(1/x)$ with $\Fc$, $\fc(x)$, and $\fu(x)$, but this time we declare the monovariant to be $M=\Fc(\iprof_S)+(1-r)\Fc(\greed)$, where $r\in (0,1)$ will be chosen such that the monovariant is nonincreasing. After the greedy step, $\Fc(\greed)$ increases by $\fc(p(1))$ and $\Fc(\iprof_S)$ increases by at most 
\begin{equation*}
-\fc(p(1))+\begin{cases}
(q(1)-p(1))\paren{\fu(q(1)-p(1))-\fu(p(1)} & q(1)<2p(1) \\
0 & \text{otherwise}
\end{cases}.
\end{equation*}

So $M$ increases by at most
\begin{align*}
&\fc(p(1))+\max_{p(1)<q(1)<2p(1)}(q(1)-p(1))(\fu(q(1)-p(1))-\fu(p)) + (1-r) \fc(p(1))\\
&= -r\fc(p(1))+\max_{0<a<p(1)}a(\fu(a)-\fu(p(1))).
\end{align*}
Now it is clear that we must choose
\begin{equation*}
r\ge \max_{0<a<p(1)}\frac{a(\fu(a)-\fu(p(1)))}{\fc(p(1))}=\max_{0<a<p(1)}\frac{a\fu(a)-a\fu(p(1))}{p\fu(p)}=\max_{0<a<p(1)}\left[\frac{\fc(a)}{\fc(p)}-\frac{a}{p}\right],    
\end{equation*}
as desired.

\subsection{Generalization of \cref{corollary:sqrt-approx}}

\begin{corollary}

When $m=2$ and$\fc(x)=x^c$, the greedy algorithm obtains a $\paren{1-c^{1/(1-c)}(1/c-1)}^{-1}$-multiplicative approximation. 

\end{corollary}

Note that for $c=\frac{1}{2}$ this is consistent with \cref{corollary:sqrt-approx}.

\begin{proof}
To apply \cref{thm:mult-approx-2} we need to compute $r=\max_{0<q<p\le 1}(q/p)^c-q/p$. Letting $t=q/p\in (0,1)$, we find 

\begin{equation*}
r=t^c-t 
\end{equation*}
\begin{equation*}
\frac{dr}{dt}=ct^{c-1}-1.
\end{equation*}
As $\frac{dr}{dt}$ is a decreasing function of $t$, we find that $r$ is maximized when $t=c^{1/(1-c)}$. Then $r=t(t^{c-1}-1)=c^{1/(1-c)}(1/c-1)$. By \cref{thm:mult-approx-2}, the greedy algorithm obtains a $(1-r)^{-1}$-multiplicative approximation, done.

\end{proof}

\subsection{Proof of \cref{thm:mult-approx-general}}

We use the same $C_S$ and $\rmassadv_S$ as in the proof of \Cref{thm:1.22}. Mirroring the steps of the proof starting from \cref{eq:multicouple-step1}, the cost of the solution found by the greedy algorithm is bounded above by
\begin{align}
\int_0^1&\fu(C_S(x))\, dx=\int_0^1 {\rmassadv_S}'(y) \fu(y)\, dy\nonumber\\
&= {\rmassadv_S}(y)\fu(y)\Big|_{y=0}^1+\int_0^1\rmassadv_S(y)(-\fu'(y))\,dy\nonumber\\
&=\fc(1)+\int_0^1\rmassadv_S(y)(-\fu'(y))\,dy\\
&=\fc(1)+\Big(\int_0^1 \frac{\iprof_S(y)+\iprof_S(2y)}{2}(-\fu'(y))+(-\fu'(y))y\int_{2y}^1\iprof'_S(t)/t\, dt \, dy\Big) \nonumber\\
&=\frac{1}{2}\paren{\fc(1)+\int_0^1 \iprof_S(y)(-\fu'(y))\, dy}+\frac{1}{2}\paren{\fc(1)+\int_0^{1}\iprof_S(2y)(-\fu'(y))\,dy}\nonumber\\
&\phantom{=}+\int_{0}^1\frac{\iprof_S'(t)}{t}\cdot \int_{0}^{t/2}-\fu'(y)y\, dy\,dt.\label{align:last-summation}
\end{align}

It remains to show that each of the three summands is bounded above by a multiple of $\Fc(\iprof_S)$.

\begin{enumerate}
    \item Using \cref{eq:fc-def-1}, twice the first summand equals $\Fc(\iprof_S)$.
    
    \item Twice the second summand is
    \begin{align*}
    \fc(1)+\int_0^{1}\iprof_S(2y)(-\fu'(y))
    &=\int_{0}^1\iprof_S'(y)\fu(y/2)\, dy\\
    &=2^{1-c}\int_{0}^1\iprof_S'(y)\fu(y)\, dy\\
    &=2^{1-c}\Fc(\iprof_S),
    \end{align*}
    where the first equality holds using integration by parts, the second by assuming $\fu(y)=y^{c-1}$, and the third by \cref{eq:fc-def-2}.

    \item First we can write 
    \begin{equation*}
        \int_0^{t/2}-\fu'(y)y\, dy=\int_{0}^{t/2}(1-c)y^{c-1}\,dy=\frac{1-c}{c}y^c\Big |_0^{t/2}=\frac{1-c}{c2^c}\fc(t).
    \end{equation*}
    So using \cref{eq:fc-def-2}, the third summand simplifies to 
    \begin{equation*}
    \frac{1-c}{c2^c}\int_0^1\iprof'(t)\fu(t)\, dt=\frac{1-c}{c2^c}\Fc(\iprof_S).
    \end{equation*}
\end{enumerate}

Putting these three summands together gives 
\begin{equation*}
(\ref{align:last-summation})=\paren{\frac{1}{2}+\frac{1}{2^c}+\frac{1-c}{c2^c}}\Fc(\iprof_S)=\paren{\frac{1}{2}+\frac{1}{c2^c}}\Fc(\iprof_S).
\end{equation*}

\begin{remark}

For $\fc(x)=x\lg(1/x)$, we were able to show an additive rather than a multiplicative guarantee because twice the second summand was bounded above by $\Fc(\iprof_S)$ plus a constant, and the third summand was bounded above by a constant.

\end{remark}

\begin{corollary}\label{corollary:general-mult-guarantee}

For general $\fc$, we can use the method in the above proof to show that the greedy algorithm achieves a $(1+a/2+b)$ multiplicative guarantee if $a$ and $b$ are constants such that $\frac{\fu(y/2)}{\fu(y)}\le a$ and $\frac{\int_0^{t/2}-\fu'(y)y\, dy}{\fc(t)}\le b$ for all values of $y\in (0,1]$ and $t\in (0,1]$, respectively.

\end{corollary}

\begin{corollary}
For any $a<2$ such that $\frac{\fu(y/2)}{\fu(y)}\le a$ for all $y\in (0,1]$, the greedy algorithm achieves an $O\paren{\frac{1}{1-a/2}}$-multiplicative guarantee.
\end{corollary}

\begin{proof}
First, using integration by parts, $\int_0^{t/2}-\fu'(y)y\, dy=-\fc(t/2)+\int_0^{t/2}\fu(y)\, dy$. The latter summand can be rewritten as $\sum_{j=1}^{\infty}\int_{t/2^{j+1}}^{t/2^j}\fu(y)\, dy$, which in turn is bounded above by a geometric series with ratio $a/2$ and leading term $\int_{t/4}^{t/2}\fu(y)\, dy\le a^2t/4\cdot \fu(t)\le O\paren{\fc(t)}$. Thus we can take $b\le O\paren{\frac{1}{1-a/2}}$ in \cref{corollary:general-mult-guarantee}.
\end{proof}

\vfill

\end{document}